\documentclass[one column,11pt,draftcls]{IEEEtran}

\usepackage{amsmath,stackrel}
\usepackage{amssymb}
\usepackage{amsthm}
\usepackage{arydshln}
\usepackage{graphicx}
\usepackage{adjustbox}
\usepackage{caption}
\usepackage{subcaption}

\usepackage{floatflt}

\newtheorem{theorem}{Theorem}
\newtheorem{lemma}[theorem]{Lemma}

\newtheorem{example}[theorem]{Example}

\newtheorem{remark}[theorem]{Remark}

\newtheorem{assumption}[theorem]{Assumption}

\newtheorem{corollary}[theorem]{Corollary}

\newcommand{{\Prob}}{\mathbb P}

\newcommand{\mypar}[1]{\vspace{0.03in}\noindent{\bf #1.}}

\begin{document}

\title{Distributed inference over directed networks: Performance limits and optimal design}

\author{Dragana Bajovi\'c, Jos\'e M. F. Moura, Jo\~{a}o Xavier, Bruno Sinopoli
\thanks{D. Bajovi\'c is with University of Novi Sad, BioSense Center, Novi Sad, Serbia.
 Jos\'e M. F. Moura and Bruno Sinopoli are with Department of Electrical and Computer Engineering, Carnegie Mellon University, Pitssburgh, PA, USA.
 Jo\~{a}o Xavier is with the Institute for Systems and Robotics, Instituto Superior T\'ecnico, University of Lisbon, Lisbon, Portugal. Authors' e-mails: dbajovic@uns.ac.rs, [moura, brunos]@ece.cmu.edu, jxavier@isr.ist.utl.pt.}}

\maketitle

\begin{abstract}
We find large deviations rates for consensus-based distributed inference for directed networks. When the topology is deterministic, we establish the large deviations principle and find exactly the corresponding rate function, equal at all nodes. We show that the dependence of the rate function on the stochastic weight matrix associated with the network is fully captured by its left eigenvector corresponding to the unit eigenvalue. Further, when the sensors' observations are Gaussian, the rate function admits a closed form expression. Motivated by these observations, we formulate the optimal network design problem of finding the left eigenvector which achieves the highest value of the rate function, for a given target accuracy. This eigenvector therefore minimizes the time that the inference algorithm needs to reach the desired accuracy. For Gaussian observations, we show that the network design problem can be formulated as a semidefinite (convex) program, and hence can be solved efficiently. When observations are identically distributed across agents, the system exhibits an interesting property: the graph of the rate function always lies between the graphs of the rate function of an isolated node and the rate function of a fusion center that has access to all observations. We prove that this fundamental property holds even when the topology and the associated system matrices change randomly over time, with arbitrary distribution. Due to generality of its assumptions, the latter result requires more subtle techniques than the standard large deviations tools, contributing to the general theory of large deviations.
\end{abstract}
%

\maketitle \thispagestyle{empty} \maketitle

\section{Introduction}
\label{sec-Intro}
The field of wireless sensor networks (WSN) has significantly evolved since its beginnings about two decades ago. Starting from wildlife monitoring, smart housing, and building and infrastructure surveillance~\cite{Aki02-SensorNetsSurvey}, the applications of WSNs have grown both in diversity and in scale. They now include monitoring and control of some highly complex large scale systems, such as vehicular networks and electric power grids. One of the important emerging trends in this field are also networks consisting of thousands of very small and simple sensing devices, such as microrobots~\cite{Abott07-MicroRobots} and nano-networks~\cite{Aki10-NanoNetworks}.

Due to the increased complexity and scale of WSNs, there has been significant interest recently in algorithms that process network information using local communications only~\cite{Barbarossa2013,Leonard2015,Cetin06distributedfusion}. A representative of this class of algorithms is the consensus algorithm~\cite{Stankovic-Estimation,running-consensus,GaussianDD}. With consensus algorithms, each agent maintains over iterations an estimate of the quantity of interest and over time it communicates the estimate to its immediate neighbors. In addition, intertwined with local communications are local agents' innovations, where agents collect new measurements and incorporate them in an iterative fashion in their current estimates. Algorithms of this form referred to as consensus+innovations~\cite{SoummyaIT12} possess several desirable features, including scalability and simplicity of implementation. Further, they are robust to structural changes in the system, such as node failures and intermittent communications, which are typical for complex systems consisting of many structurally simple devices. In terms of applications, consensus algorithms have been applied in various different contexts: distributed Kalman filtering~\cite{Olfati-Saber-KalmanEmbeddedConsensus05,CarliDKF08}, distributed detection~\cite{GaussianDD,Non-Gaussian-DD,running-consensus,stankovic-change-detection} and parameter estimation~\cite{Stankovic-Estimation,MateosConsensusDistributedLMS2009,SoummyaIT12}, distributed learning~\cite{SayedSparseLearning13}, and tracking~\cite{RahmanDistributedTracking07}.

In this paper, we study large deviations performance of consensus algorithms when the underlying network is directed. This complements the existing work that usually studies asymptotic variance or asymptotic normality~\cite{SoummyaIT12,SoummyaIT14}. Our goal is to compute (or characterize--when exact computation is not possible) the rates at which the local nodes' estimates converge to the desirable values (e.g., the vector of true parameters that are being estimated). To explain the relevance of large deviations performance, consider, for example, a binary hypothesis testing problem in a WSN. In this context, the rates of large deviations correspond to error exponents, i.e., they provide answers to how fast the error probabilities -- false alarm, missed detection, or total error probability decay with time. In the context of estimation,  large deviation rates provide estimates of times to reach a desired accuracy region around the true parameter that the local estimates converge to. Naturally, the higher the rate of a node, the better is the decision or estimation produced by that node at a given time. One particular goal of this paper is to provide answers to questions such as: ``How much faster a node in a network filters out the estimation noise compared to a node that operates alone?''

\mypar{Contributions} We consider both cases when the local nodes' interactions are deterministic and when they are random, where the local interactions are captured by associated stochastic system matrices\footnote{With a stochastic matrix, rows sum to one, and all the entries are nonnegative.}. For the deterministic case, we prove the large deviation principle at each node, and we find the corresponding rate function, equal at all nodes. We prove that its dependence on the (stochastic) system matrix $A$ is fully captured by the left eigenvector $a$ of $A$ associated with the eigenvalue one, i.e., the left Perron vector of~$A$. When the observations are Gaussian-- independent, but non-identically distributed, we find a closed form expression for the rate function.
Motivated by the fact that the rate function strongly depends on the eigenvector $a$, we formulate the following network design problem. For a given accuracy region, find the optimal vector $a$ that maximizes the value of the rate function on this fixed region. We further show that for Gaussian observations with equal means (but different covariance matrices), this problem can be formulated as a semidefinite program (SDP) and thus can be solved efficiently. Simulation examples demonstrate that the optimized system significantly outperforms the system with the uniform left eigenvector $a$ that, in a sense, equally ``weighs'' all of the nodes' estimates. Finally, considering the special case when the observations are independent and identically distributed (i.i.d.), we reveal a very interesting property: the rate function, independently of the choice of $A$, always lies between the rate function of an isolated node and the rate function of a fusion center. Intuitively, this means that the distributed system is always better than an isolated node, and that, on the other hand, can never beat the performance of a fusion center.  Moreover, we prove that this fundamental property holds with random system matrices of arbitrary distribution (including, e.g., temporal dependencies), as long as they are independent from the observations. Due to the generality of the assumptions, the proof of this result requires much more sophisticated techniques than the deterministic case, which improve over the state of the art large deviations techniques and hence constitute a contribution of its own.

\mypar{Related work} Large deviations asymptotic performance of consensus+innovations algorithms has been previously studied in~\cite{GaussianDD},\cite{Non-Gaussian-DD},\cite{Soummya-LDRiccati14},\cite{BMMSayed14}, and~\cite{TaraISIT2014}. Reference~\cite{Soummya-LDRiccati14} studies large deviations of the stochastic Riccatti equation for the distributed Kalman filter, and it provides an upper and a lower bound for the large deviations rate function. Reference~\cite{BMMSayed14} considers a consensus based distributed detection with constant learning step. They show that the local decision statistics satisfy the large deviations principle and characterize the corresponding rate function. Reference~\cite{TaraISIT2014} studies belief formations in social networks and characterizes error exponents (Kullback-Leibler divergences) for the distributed multiple hypothesis testing problem. In our previous work~\cite{GaussianDD},\cite{Non-Gaussian-DD}, we considered the case of i.i.d. networks, where each topology realization is symmetric. Under this model, reference~\cite{GaussianDD} finds an upper and a lower bound for the rate function when the observations are Gaussian, and reference~\cite{Non-Gaussian-DD} extends the results of~\cite{GaussianDD} to arbitrary distributions of sensor observations. In this work, we go beyond these results in several important directions. First, we study here \emph{directed} random networks, and, furthermore, we make no restrictions on the distribution of the system matrices; in particular, we allow for their arbitrary \emph{time correlations}. Second, when the system matrices are deterministic, asymmetric, we fully characterize the rate function and show that it is amenable to optimization.

\mypar{Notation} For arbitrary $d\in \mathbb N=\left\{1,2,...\right\}$, we denote by $0_d$ the $d$-dimensional vector of all zeros; by $1_d$ the $d$-dimensional vector of all ones; by $e_i$ the $i$-th canonical vector of $\mathbb R^d$ (that has value one on the $i$-th entry and the remaining entries are zero); by $I_d$ the $d$-dimensional identity matrix; by $J_d$ the $d\times d$ matrix whose all entries equal $1/d$. For a matrix $A$, we let $[A]_{ij}$ and $A_{ij}$ denote its $i,j$ entry and  for a vector $a\in \mathbb R^d$, we denote its $i$-th entry by $a_i$, $i,j=1,...,d$. For a function $f:\mathbb R^d\mapsto \mathbb R$, we denote its domain by $\mathcal D_f=\left\{ x\in \mathbb R^d: -\infty <f(x)<+\infty \right\}$; the subdifferential (gradient, when $f$ is differentiable) of $f$ at a point $x$ by $\partial f(x)$ ($\nabla f(x)$); $\log$ denotes the natural logarithm; for two sequences $f_t$ and $g_t$ that are asymptotically equal at the logarithmic scale, $\lim_{t\rightarrow +\infty} \log f_t/\log g_t=1$, we shortly write $f_t \stackrel {\star}{\sim} g_t$. For $N\in \mathbb N$, we denote by $\Delta^{N-1}$ the probability simplex in $\mathbb R^N$ and by $\alpha$ the generic element of this set: $\Delta^{N-1}=\left\{\alpha \in \mathbb R^N: \alpha_i \geq 0, \sum_{i=1}^N \alpha_i=1\right\}$. We let $\lambda_{\max}$ and $\lambda_2$, respectively, denote the maximal and the second largest (in modulus) eigenvalue of a square matrix; $\dagger$ denotes the pseudoinverse of a square matrix; and $\|\cdot\|$ denotes the spectral norm. For a matrix $S\in \mathbb R^{N\times N}$, we let $\mathcal R(S)$ denote the range of $S$, $\mathcal R(S)=\left\{Sx: x\in \mathbb R^N\right\}$, and for $N$ square matrices $S_1,...,S_N$, we let $\mathrm{diag}\left\{S_1,...,S_N\right\}$ denote the block-diagonal matrix whose $i$th block is $S_i$, for $i=1,...,N$. An open Euclidean ball in $\mathbb R^d$ of radius $\rho$ and centered at $x$ is denoted by $B_{x}(\rho)$; the closure, the interior, the boundary, and the complement of an arbitrary set $D\subseteq \mathbb R^d$ are respectively denoted by $\overline D$, $D^{\mathrm{o}}$, $\partial D$, and $D^{\mathrm{c}}$; $\mathcal B(\mathbb R^d)$ denotes the Borel sigma algebra on $\mathbb R^d$; $\Omega$ denotes the probability space and $\omega$ denotes an element of $\Omega$; $\mathbb P$ and $\mathbb E$ denote the probability and the expectation operator; $\mathcal N(m,S)$ denotes Gaussian distribution with mean vector $m$ and covariance matrix $S$.

\mypar{Paper organization} In Section~\ref{sec-Setup} we present the system model and formulate the problem that we study. In Section~\ref{sec-Pre} we give preliminaries. Section~\ref{sec-rate-fcns-deterministic} presents our results for the deterministic case. Using the results of Section~\ref{sec-rate-fcns-deterministic}, Section~\ref{sec-network-design} formulates the network design problem and solves it for the case of Gaussian observations with equal means. Section~\ref{sec-Generic} presents the fundamental bounds on the rate function for the generic case, when system matrices are random; proofs of this result are given in Subsections~\ref{sec-Upper-Proof} and~\ref{sec-Lower-Proof}. Simulation results are presented in Section~\ref{sec-Simul}, and the conclusion is given in Section~\ref{sec-Concl}.

\section{Problem setup}
\label{sec-Setup}
This section explains the system model and the distributed inference algorithm that we study.

\mypar{Network observations} Suppose that we have $N$ geographically distributed agents (e.g., sensors, robots, humans) that monitor and collect observations about their environment. We denote the set of agents by $V=\{1,2,\ldots,N\}$ such that $i\in V$ denotes the $i$-th agent. At each new time instant $t=1,2,...$, each agent produces a $d$-dimensional observation vector. We denote by $Z_{i,t}\in \mathbb R^d$ the observation vector of agent $i$ at time $t$, where we assume that the measurements are made synchronously across all agents. We denote by $m_i$ the expected value of observations at node $i$, $m_i=\mathbb E\left[Z_{i,t}\right]$ (constant for all $t$). 

\mypar{Inter-agent communication} We assume that a direct communication is possible only between a subset of agents' pairs, e.g., the agents that are close enough to each other. (For instance, in a WSN, communication links are established only between sensors that lie within a certain, predefined distance $r$ from each other.) We model the possible inter-agent communications via a directed graph $\widehat G=(V,\widehat E)$, where set $\widehat E\subseteq V\times V$ collects all possible (directed) communication links, i.e., all pairs $(j,i)$ such that agent $i$ can receive messages from agent $j$ in a single hop manner. The links in $\widehat E$ should be understood only as potential communication channels. In other words, at a certain time $t$, agent $j$ may decide whether to send or not send a message to agent $i$. Also, in the case a message from $j$ to $i$ was sent, its reception at $i$ could be unsuccessful due to imperfect channel effects (e.g., fading). For any link $(j,i)\in \widehat E$, we say that $(j,i)$ is active at time $t$ if at time $t$ a message is sent from $j$ and successfully received at $i$. We let $E_t$ denote the set of all active links at time $t$. Accordingly, the neighborhood of node $i$ at time $t$ is $O_{i,t} = \left\{ j: (j,i)\in E_t \right\}$, that is, $O_{i,t}$ is the set of all active links at time $t$ that are pointing to $i$; for any $j\in O_{i,t}$, we say that $j$ is an active neighbor of $i$. Finally, we denote by $G_t = (V, E_t)$ the graph realization at time $t$.

\mypar{Consensus+innovations based distributed inference} The distributed inference algorithm that we study operates as follows. Each node, over time, maintains a $d$-dimensional vector that serves as the node's estimate on the state of nature. The estimate of node $i$ at time $t$ is denoted by $X_{i,t}$, and we also refer to it as the state of node $i$. The estimates (states) are continuously improved over time twofold. First, each agent $i$ incorporates its new observation $Z_{i,t}$ into its current state with the weight $1/t$ and forms an intermediate state update; subsequently, it transmits the intermediate state to (a subset of) its neighbors. Finally, agent $i$ forms a convex combination (weighted average) of its own and its active neighbors' intermediate states, with the coefficients $\left\{W_{ij,t}:j\in O_{i,t}\right\}$, $i\in V$. Mathematically, the state update of agent $i$ is:
\begin{align}
\label{alg-D-INF}
X_{i,t} =\sum_{ j \in O_{i,t}} W_{ij,t} \left(\frac{t-1}{t} X_{j,t-1} +\frac{1}{t} Z_{j,t}\right),
\end{align}
with the initialization $X_{i,0}=0_d$. To derive a more compact representation, collect for each $t$ the agents' weights $W_{ij,t}$ in an $N \times N$ matrix $W_t$ as follows: for any pair $(j,i)\in \widehat E$ that satisfies $j\in O_{i,t}$, $[W_t]_{ij}$ is assigned the value $W_{ij,t}$, and equals zero otherwise, and for any $i\in V$, $[W_t]_{ii}=1-\sum_{j\in O_{i,t}}[W_t]_{ij}$. We refer to matrix $W_t$ as the weight matrix. Due to the fact that $\{W_{ij,t}:j\in O_{ij,t}\}$ form a convex combination, $W_t$ is stochastic for any $t$. 
Further, let $\Phi(t,s)$, for $t\geq 1$  and $t \geq  s\geq  1$ be defined as $\Phi(t,s)=W_t\cdots W_{s}$, for $1 \leq s\leq t$. From~\eqref{alg-D-INF}, we obtain:
\begin{equation}
\label{alg-D-INF-compact}
X_{i,t}=\frac{1}{t}\sum_{s=1}^t\sum_{j=1}^N \left[\Phi(t,s)\right]_{ij} Z_{j,s}.
\end{equation}
Algorithms of form~\eqref{alg-D-INF} and~\eqref{alg-D-INF-compact} have been previously studied, e.g., in~\cite{Stankovic-Estimation},\cite{running-consensus}, and~\cite{GaussianDD}.

We now state our assumptions on the weight matrices and the agents' observations.
\begin{assumption}[Network and observation model]
\label{ass-network-and-observation-model}
$\phantom{enforce}$
\begin{enumerate}
\item Observations $Z_{i,t}$, $i=1,\ldots,N$, $t=1,2,\ldots$ are independent both across nodes and over time;
\item For each agent $i$, $Z_{i,t}$, $t=1,2,...$ are identically distributed;
\item Quantities $W_t$ and $Z_{i,s}$ are independent for all $i$, $s$, $t$.
\end{enumerate}
\end{assumption}
The model above is very general. In particular, in terms of the agents' interactions, it allows for directed topologies and asymmetric weight matrices, and it also allows for time dependencies between the weight matrices; directed topologies and temporal dependencies are cases that are much less studied in the literature. 
In terms of observations, we remark that the model above allows for non-identically distributed observations.

We next introduce the rates of large deviations and motivate their use for performance characterization of algorithm~\eqref{alg-D-INF}.

\mypar{Rates of large deviations at individual agents}
Suppose that, for some $i$, $X_{i,t}$ converges almost surely (a.s.) to a deterministic vector $\theta\in \mathbb R^d$, e.g., the vector of $d$ parameters that the system wishes to estimate. 
 In many scenarios, it is of interest to determine at what rate this convergence occurs. To explain why this is important, suppose that we wish to determine $\theta$ up to a certain accuracy defined by the accuracy region $C\subseteq \mathbb R^d$, where $\theta\in C$. Let $T_i$ denote the time interval after which $X_{i,t}$ belongs to $C$ with a prescribed, high probability, say $0.97$. For convenience, define also the complement of $C$, $D=\mathbb R^d\setminus C$, usually called the deviation set. Since $X_{i,t}$ converges a.s. to $\theta,$ we know that the probability that $X_{i,t}$ remains outside of $C$, $\mathbb P\left(X_{i,t}\in D\right)$, vanishes as $t\rightarrow +\infty$. The question that we ask then is how fast this probability vanishes with time. It turns out that in many scenarios this convergence is exponential (see~\cite{Non-Gaussian-DD} for the scalar, $d=1$ case). That is:
\begin{equation}
\mathbb P\left(X_{i,t}\in D\right)\stackrel {\star}{\sim} e^{-t I_{i}(D)},
\end{equation}
for a certain function $I_{i}$, where, we recall, $\stackrel {\star}{\sim}$ means that the two functions are asymptotically equal at the logarithmic scale. Function $I_i:\mathcal B\left(\mathbb R^d\right)\,\mapsto\,\mathbb R^{+}$ is usually called the rate function. Relating $I_i$ with time $T_i$, we see that $T_i$ can be approximately computed as
\begin{equation}
\label{eq-estimate-T-via-I}
T_i\approx - \frac{\log (1-0,97)}{I_i(D)}.
\end{equation}
The quality of the approximation in~\eqref{eq-estimate-T-via-I} improves for higher accuracies (i.e., smaller region $C$ around $\theta$). In the context of, e.g., Neyman-Pearson hypothesis testing, rates $I_i$ directly correspond to error exponents: taking, for example $D$ to be the false alarm region $[0,+\infty)$ under $H_0$, $I_i(D)$ gives the error exponent of the false alarm probability at sensor $i$. The problem that we address in this paper is finding the rate functions $I_i$, $i\in V$:
\begin{equation}
\lim_{t\rightarrow +\infty} \,-\frac{1}{t}\,\log\,\mathbb P\left(X_{i,t}\in D\right) = I_i(D),
\end{equation}
whenever the limit above exists for any set $D\in \mathcal B(\mathbb R^d)$. For further details on the use of large deviations rate functions in probabilistic inference, we refer the reader to~\cite{Cover91},\cite{Chernoff52},\cite{LD-for-M-estimators-06}.


\section{Preliminaries}
\label{sec-Pre}
Before we start our analysis, we first review in Subsection~\ref{subsec-LD-preliminaries} basic large deviations concepts and tools. Subsection~\ref{subsec-isolation-and-fusion} then provides our intermediate results on the large deviations principle and the corresponding rate functions of an isolated agent and a fusion node.
\subsection{Large deviations preliminaries}
\label{subsec-LD-preliminaries}
We define the large deviations principle and introduce, for each $i$, the logarithmic moment generating function of observations~$Z_{i,t}$. We then define the conjugate of a function and state some important properties of log-moment generating functions and their conjugates in general, and in our particular setup as well.

\mypar{Large deviations principle} A rate function is any function that is lower semi-continuous, or equivalently, that has closed sublevel sets.  A sequence of random variables $\widehat Z_t\in \mathbb R^d$ is said to satisfy the large deviations principle (LDP) with rate function $\widehat I$ if for any measurable set $D\in \mathcal B(\mathbb R^d)$ it holds that
\begin{align}
\label{def-LDP}
- \inf_{x \in D^{\mathrm{o}}} \widehat I(x)\, \leq \liminf_{t\rightarrow +\infty}\, \frac{1}{t}\,\mathbb P\left( \widehat Z_t \in D \right)
 \leq\, \limsup_{t\rightarrow +\infty}\, \frac{1}{t}\,\mathbb P\left( \widehat Z_t \in D \right) \leq - \inf_{x\in \overline D} \widehat I(x).
\end{align}
Essentially, what the large deviations principle tells is that, for any (nice enough) set $D$, probabilities that $\widehat Z_t$ belongs to $D$ decay with $t$ exponentially, with the rate equal to $\widehat I(D)=\inf_{x\in D} \widehat I(x)$. One of the key objects in proving the large deviations principle and computing the rate function in general (see Cram\'er's and G\"{a}rtner-Ellis theorem~\cite{Cramers-article},\cite{DemboZeitouni}) are the log-moment generating function and its conjugate, which we introduce next.

\mypar{Log-moment generating function of observations~$Z_{i,t}$}
The log-moment generating function $\Lambda_i:\,{\mathbb R}^d \rightarrow \mathbb R \cup \{+\infty\}$ corresponding to $Z_{i,t}$ is given by:
\begin{equation}
\label{def-Lmgf}
\Lambda_i(\lambda)=\log \mathbb E\left[ e^{\lambda^\top Z_{i,t}}\right],\:\:\mathrm{for\:\:}\lambda \in \mathbb R^d.
\end{equation}
%
For the special case when all the agents' observations are identically distributed, we let $\Lambda$ denote the corresponding log-moment generating function, $\Lambda\equiv \Lambda_i$, for any $i$.

The second key object of interest in our analysis is the conjugate of a log-moment generating function. Let $\widehat \Lambda$ be the log-moment generating function of a $d$-dimensional random vector $\widehat Z$. Then, the conjugate, or the Fenchel-Legendre transform, of $\widehat \Lambda$ is given by
\begin{equation}
\label{def-Conjugate}
\widehat I(x)=\sup_{\lambda\in \mathbb R^d} x^\top \lambda - \widehat \Lambda(\lambda),\:\:\mathrm{for\:\:}x\in \mathbb R^d.
\end{equation}

When $Z_{i,t}$ are i.i.d., we will denote by $I$ the conjugate of $\Lambda$. To illustrate how to compute~$\Lambda$ and $I$, we consider the case when $Z_{i,t}$ is a discrete random vector, i.e., when the agents' measurements are quantized.
\begin{example}[Quantized observations]
\label{ex-lmgf-discrete}
Suppose that the agents' observations $Z_{i,t}$ are i.i.d., discrete random vectors, taking values in the set $\mathbb A= \left\{a_1,...,a_L\right\}$, according to the probability mass function $p=(p_1,...,p_L)$, where $a_l\in \mathbb R^d$ for $l=1,...,L$. For any $\lambda \in \mathbb R^d$, the value $\Lambda (\lambda)$ is then computed by
\begin{equation}
\label{lmgf-for-discrete}
\Lambda(\lambda)=\log \left(\sum_{l=1}^L p_l e^{\lambda^\top a_l} \right).
\end{equation}
It can be seen that the function $\Lambda$ in~\eqref{lmgf-for-discrete} is finite on the whole space, i.e., $\mathcal D_{\Lambda}=\mathbb R^d$. Also, for the special case when $a_l=e_l$, the conjugate of~$\Lambda$ can be shown to be the relative entropy with respect to $p$, given by~\cite[p. 41]{DemboZeitouni}:
\begin{equation}
\label{conjugate-for-discrete}
I(x)=\sum_{l=1}^d x_l \log \frac{x_l}{p_l},
\end{equation}
for any $x\in \Delta^{d-1}$, and equals $+\infty$, otherwise.
\end{example}
\begin{example}[Gaussian observations]
\label{ex-lmgf-Gauss}
It can be shown by simple algebraic manipulations that when $Z_{i,t}$ is i.i.d., Gaussian, with mean value~$m$ and covariance matrix~$S$, the log-moment generating function~$\Lambda$ and its conjugate~$I$ are both quadratic and given, respectively, by~\cite{DemboZeitouni}: \[\Lambda(\lambda)=  m^\top \lambda +  \frac{1}{2} \lambda^\top S \lambda,\,\,I(x)=\frac{1}{2}(x-m)^\top S^{-1}(x-m).\]
\end{example}
To simplify our analysis, we make the following assumption.
\begin{assumption}
\label{ass-finite-at-all-points}
$\mathcal D_{\Lambda_i}= {\mathbb R}^d$, i.e., $\Lambda_i(\lambda)<+\infty$ for all $\lambda \in \mathbb R^d$, for each $i$.
\end{assumption}
Assumption~\ref{ass-finite-at-all-points} holds for arbitrary Gaussian and discrete random vectors, and also for many other commonly used distributions; we refer the reader to~\cite{Non-Gaussian-DD} for examples of random vectors beyond Examples~\ref{ex-lmgf-discrete} and~\ref{ex-lmgf-Gauss} that have a finite log-moment generating function.

\mypar{Properties of log-moment generating functions and their conjugates}
For future reference, we list the properties that an arbitrary log-moment generating function $\widehat \Lambda$ and its conjugate $\widehat I$ satisfy; proofs can be found in~\cite[p.8]{Hollander} and~\cite[p.27, 35]{DemboZeitouni}.%
\begin{lemma}[Properties of a log-moment generating function and its conjugate] Consider the log-moment generating function $\widehat \Lambda$ and its conjugate $\widehat I$, associated with an arbitrary $d$-dimensional random vector $\widehat Z$. Let $\theta =\mathbb E[\widehat Z]$. Then:
\begin{enumerate}
\item
\label{lemma-properties-of-lmgf}
function~$\widehat \Lambda$ satisfies:
\begin{enumerate}
\item \label{Lambda-at-0} $\widehat \Lambda(0)=0$ and $\nabla \widehat \Lambda(0)=\theta$, when $0\in \mathcal D_{\widehat \Lambda}^{\mathrm{o}}$;
\item \label{Lambda-is-cvx} $\widehat \Lambda(\cdot)$ is lower semi-continuous and convex;
\item \label{Lambda-is-smooth} $\widehat \Lambda(\cdot)$ is $C^{\infty}$ on $\mathcal D_{\widehat \Lambda}^{\mathrm{o}}$;
\end{enumerate}
\item
\label{lemma-properties-of-the-conjugate}
and function~$\widehat  I$ satisfies:
\begin{enumerate}
\item $\widehat I$ is nonnegative and $\widehat I\left(\theta\right)=0$;
\item \label{part-I-qualifies} $\widehat I$ is lower semi-continuous and convex;
\item if $0\in \mathcal D_{\widehat \Lambda}^{\mathrm{o}}$, then $\widehat I$ has compact level sets.
\item $\widehat I$ is differentiable on $\mathcal D_{\widehat I}^{\mathrm{o}}$.
\end{enumerate}
\end{enumerate}
\end{lemma}
%

We end this subsection by stating a simple but important property of the log-moment generating function that follows from its convexity and zero value at the origin. We note that the right-hand side of inequality~\eqref{eq-simple-lemma} was previously proven in~\cite{Non-Gaussian-DD} (for the case $d=1$).
\begin{lemma}
\label{simple-lemma}
Let $\widehat \Lambda$ be an arbitrary log-moment generating function. For any $\alpha \in \Delta^{N-1}$ and $\lambda \in \mathbb R^d$,
\begin{equation}
\label{eq-simple-lemma}
N \widehat \Lambda \left(\frac{1}{N} \lambda\right)\leq \sum_{i=1}^N \widehat \Lambda(\alpha_i \lambda) \leq \widehat \Lambda(\lambda).
\end{equation}
\end{lemma}
\begin{proof}
We first prove the right-hand side inequality in~\eqref{eq-simple-lemma}. (The proof is analogous to the proof of the same inequality for the special case $d=1$~\cite{Non-Gaussian-DD}; for completeness, we provide the proof here.) Fix $\varsigma \in [0,1]$. Then, by convexity of $\widehat \Lambda$ and the fact that $\widehat \Lambda(0)=0$, we have
\begin{align*}
\widehat \Lambda(\varsigma \lambda)= \widehat \Lambda(\varsigma\, \lambda + (1-\varsigma)\,0)\,\leq\, \varsigma\, \widehat \Lambda(\lambda) + (1-\varsigma)\,\widehat \Lambda(0) = \varsigma \widehat \Lambda(\lambda).
\end{align*}
Now, fix an arbitrary $\alpha \in \Delta^{N-1}$. Applying the preceding inequality for $\varsigma = \alpha_i$, for $i=1,...,N$, yields the claim by summing out the resulting left and right hand sides.

To prove the left hand side inequality in~\eqref{eq-simple-lemma}, consider the function $g_{\lambda}: \mathbb R^N \mapsto \mathbb R$,
$g_{\lambda}(\beta)=\sum_{i=1}^N \widehat \Lambda(\beta_i\lambda)$, for $\beta \in \mathbb R^N$. We prove the claim if we show that the minimum of
$g_{\lambda}$ over the unit simplex $\Delta^{N-1}$ is attained at $1/N\, 1_N=(1/N,\ldots,1/N)\in \Delta^{N-1}$.
Since $g_{\lambda}$ is convex (being the sum of convex functions), it suffices to show that there exists a Lagrange multiplier $\nu \in \mathbb R$ such that the pair $(1/N\, 1_N,\nu)$ satisfies the Karush-Kuhn-Tucker~(KKT) conditions~\cite{BoydsBook}. To this end, define the Lagrangian $L(\beta,\nu)= g_{\lambda}(\beta)+ \nu (1_N^\top\, \beta -1)$, for some $\nu \in \mathbb R$, $\beta \in \mathbb R^N$. We have
\begin{equation*}
\partial_{\beta_i} L(\beta, \nu)= \lambda^\top \nabla \widehat \Lambda(\beta_i \lambda) + \nu.
\end{equation*}
Taking $\beta_i=1/N$ and $\nu= - \lambda^\top \nabla  \widehat \Lambda(1/N \lambda)$, proves the claim.
\end{proof}
\vspace{-5mm}
\subsection{Two extreme cases: isolation and fusion}
\label{subsec-isolation-and-fusion}

To set benchmarks for the performance of distributed inference~\eqref{alg-D-INF}, we consider two extreme cases of the agents' cooperation: 1) complete agent's isolation, when an agent operates alone, making inferences based on its own observations only; and 2) network-wide fusion, when each agent has access to all of the observations. Mathematically, the state of agent $i$ corresponding to these two cases are as follows: $X_{i,t}^{\mathrm{isol}}= 1/t\sum_{s=1}^t Z_{i,s}$, for $i\in V$, for the case of isolated agents obtained when in~\eqref{alg-D-INF-compact} $W_t\equiv I_d$, and $X_{t}^{\mathrm{cen}}= 1/(Nt)\sum_{s=1}^t\sum_{i=1}^N Z_{i,s}$, for the case of fusion center obtained when $W_t\equiv J_d$. Theorem~\ref{theorem-isolation-and-fusion} computes the corresponding large deviation rates, and it also asserts that, when the observations are i.i.d., the fusion-based rate scales linearly (with constant one) with the number of participating agents.

\begin{theorem}
\label{theorem-isolation-and-fusion}
Suppose that $Z_{i,t}$ are i.i.d. for all $i$ and $t$. Then,
\begin{enumerate}
\item
\label{part-isolation}
for each $i$, the sequence $X_{i,t}^{\mathrm{isol}}$ satisfies the LDP with rate function $I_i^{\mathrm{isol}}\equiv I$;
\item
\label{part-fusion}
the sequence $X_{t}^{\mathrm{cen}}$ satisfies the LDP with rate function $I^{\mathrm{cen}}\equiv NI$.
\end{enumerate}
\end{theorem}
Clearly, by the strong law of large numbers, 
 with both isolated nodes and fusion center, the corresponding states $X_{i,t}^{\mathrm{isol}}$, $i=1,...,N$, and $X_{t}^{\mathrm{cen}}$ converge a.s. to $m:=\mathbb E[Z_{i,t}]$.
\begin{proof}
Since $Z_{i,s}$ are i.i.d., part~\ref{part-isolation} follows by a direct application of Cram\'er's theorem~\cite{Cramers-article}, \cite[p.36]{DemboZeitouni}.
Turning to part~\ref{part-fusion}, note that $X_{t}^{\mathrm{cen}}$ can be written as an average of i.i.d. samples $1/N \sum_{i=1}^N Z_{i,s}$, $X_{t}^{\mathrm{cen}}= 1/t\sum_{s=1}^t 1/N \sum_{i=1}^N Z_{i,s}$. Thus, again by an application of Cram\'er's theorem~\cite{Cramers-article}, we see that to prove part~\ref{part-fusion} it suffices to show that the conjugate of the log-moment generating function of $1/N \sum_{i=1}^N Z_{i,s}$ is $NI$. Computing the log-moment generating function of $1/N \sum_{i=1}^N Z_{i,s}$ at $\lambda\in \mathbb R^d$, we obtain:
\begin{align*}
\log \mathbb E\left[ e^{\frac{1}{N} \sum_{i=1}^N \lambda^\top Z_{i,s}}\right]&= \sum_{i=1}^N \log \mathbb E\left[e^{\frac{1}{N} \lambda^\top Z_{i,s}}\right]= N \Lambda\left(\frac{\lambda}{N}\right),
\end{align*}
where in the first equality we used the fact that the $Z_{i,s}$ are independent, for fixed $s$, and in the second equality we used that they are identically distributed, with log-moment generating function $\Lambda$. Finally, simple algebraic manipulations reveal that the conjugate of $ N \Lambda(\lambda/N)$ equals $NI$: for any $x\in \mathbb R^d$
\begin{align*}
\sup_{\lambda \in \mathbb R^d} x^\top \lambda -  N \Lambda\left(\frac{\lambda}{N}\right)=
N \left(\sup_{\lambda \in \mathbb R^d} x^\top \left(\frac{\lambda}{N}\right) -  N \Lambda\left(\frac{\lambda}{N}\right)\right) = NI(x).
\end{align*}
\end{proof}

Theorem~\ref{theorem-isolation-and-fusion} asserts that the rate function of any isolated agent $i$ is $I_i^{\mathrm{isol}}\equiv I$, where $I$ is the conjugate of the log-moment generating function of its observation, whereas the rate function of the network-wide fusion is $N$ times higher, $I_i^{\mathrm{cen}}\equiv NI$. Intuitively, for the general case of algorithm~\eqref{alg-D-INF}, we expect that the rate function of a fixed agent $i$ should be between these two functions, $I$ and $NI$. It turns out that this is indeed the case -- Corollary~\ref{corollary-between-one-and-all-deterministic} proves this for deterministic matrices, and Theorem~\ref{theorem-universal-limits} later in Section~\ref{sec-Generic} confirms that this is true even for arbitrary (asymmetric) random matrices.
\section{Rate functions $I_i$ for deterministic weight matrices}
\label{sec-rate-fcns-deterministic}
This section considers deterministic weight matrices. The first result that we present, Theorem~\ref{theorem-deterministic}, computes the rate functions $I_i$ for the case when the weight matrices at all times are equal to a stochastic matrix $A$ such that $|\lambda_2(A)|<1$. (This means that the underlying network has only one initial class\footnote{An initial class of a directed graph $G$ is any communication class of $G$ that has no incoming edges~\cite{Tahbaz-Salehi08}. We also note that initial classes of $G$ correspond to essential classes of the transpose of $G$ (the graph that results from reversing the directions of edges in $G$~\cite{SenetaBook}).}, e.g.,~\cite{Kirkland09,SenetaBook}.) We then focus on the special case when all observations are Gaussian (with possibly different parameters across agents), and we calculate the rate functions in closed form. Further, we formulate the problem of optimal network design and show that it can be efficiently solved by an SDP when the observations are Gaussian.

\begin{theorem}
\label{theorem-deterministic}
Let $W_t\equiv A$ for each $t$ and let Assumptions~\ref{ass-network-and-observation-model} and~\ref{ass-finite-at-all-points} hold. Suppose that $|\lambda_2(A)|<1$  and let $a$ denote the left eigenvector of $A$ corresponding to the eigenvalue $1$. Then,
for each $i$, $X_{i,t}$, $t=1,2,...$ satisfies the LDP with the rate function $I_i\equiv \widetilde I$,
where $\widetilde I$ is the conjugate of \[\widetilde\Lambda(\lambda):=\sum_{j=1}^N \Lambda_j(a_j\lambda),\;\;\;\;\;\lambda \in \mathbb R^d.\] Moreover, for each $i$, $X_{i,t}$ converges a.s. to $\widetilde m:=\sum_{j=1}^N a_j m_j$.
\end{theorem}
\begin{proof}
To prove the first part of the theorem, we apply the G\"{a}rtner-Ellis theorem~\cite{DemboZeitouni}. Fix $i\in V$ and let $\Lambda_t(\lambda):= \frac{1}{t}\log \mathbb E\left[e^{t\lambda X_{i,t}} \right]$, for $\lambda \in \mathbb R^d$.  Using that $Z_{i,t}$ are independent and that $\Phi(t,s)=A^{t-s+1}$ are constant, we obtain
\begin{align}
\label{eq-computing-Lambda-t}
\Lambda_t(\lambda)& = \frac{1}{t}\log \mathbb E\left[e^{\lambda  \sum_{s=1}^t\sum_{j=1}^N  [\Phi(t,s)]_{ij} Z_{j,s}}   \right] \nonumber\\
&= \frac{1}{t} \sum_{s=1}^t \sum_{j=1}^N \log  \mathbb E\left[ e^{\lambda  [\Phi(t,s)]_{ij} Z_{j,s} }\right]\nonumber \\
&= \frac{1}{t} \sum_{s=1}^t \sum_{j=1}^N \Lambda_j\left([A^{t-s+1}]_{ij} \lambda\right)\nonumber \\
&= \sum_{j=1}^N \frac{1}{t} \sum_{r=1}^t \Lambda_j\left([A^{r}]_{ij} \lambda\right).
\end{align}
From $|\lambda_2(A)|<1$ we have that $A^r\rightarrow 1a^\top$ as $r\rightarrow +\infty$~\cite{MatrixAnalysis}, and, hence, for any $i$, $[A^r]_{ij}\rightarrow a_j$.
Consider now a fixed~$j$. Then, by continuity of $\Lambda_j$, $\Lambda_j\left([A^{r}]_{ij} \lambda\right)\rightarrow \Lambda_j(a_j\lambda)$, and hence the C\'esaro averages must converge to the same number:
\begin{equation*}
\lim_{t\rightarrow +\infty} \frac{1}{t} \sum_{r=1}^t \Lambda_j\left([A^{r}]_{ij} \lambda\right)= \Lambda_j(a_j\lambda).
\end{equation*}
Going back to~\eqref{eq-computing-Lambda-t} and taking the limit yields $\lim_{t\rightarrow +\infty}\Lambda_t(\lambda) = \sum_{j=1}^N \Lambda(a_j \lambda).$
Thus, conditions for applying the G\"{a}rtner-Ellis theorem are fulfilled, and thus we have that, for each $i$, $X_{i,t}$ satisfies the large deviations principle with the rate function equal to the conjugate of $\sum_{j=1}^N \Lambda_j(a_j \lambda)$.

It remains to prove that $X_{i,t}$ at each $i$ converges to $\widetilde m=\sum_{j=1}^N a_j m_j$. First, note that $\widetilde \Lambda$ is in fact the log-moment generating function of $\sum_{j=1}^N a_j Z_{j,t}$, for any $t$. This easily follows from the independence of the $Z_{j,t}$'s, for $t$ fixed:
\begin{align*}
\log\mathbb E\left[e^{\lambda^\top \sum_{j=1}^N a_j Z_{j,t} }\right]= \sum_{j=1}^N \log \mathbb E\left[ e^{a_j \lambda^\top Z_{j,t}} \right] =\sum_{j=1}^N \Lambda_j(a_j\lambda).
\end{align*}
Thus, being a log-moment generating function, $\widetilde \Lambda$ satisfies the properties given in Lemma~\ref{lemma-properties-of-lmgf}. In particular, from the lower semicontinuity and convexity of $\widetilde \Lambda$ it follows that $\widetilde \Lambda$ and $\widetilde I$ are the conjugates of each other. Invoking a classical result for conjugate functions, see, e.g., eq.~(1.4.6) on p.~222 in~\cite{Urruty}, we have:
\begin{equation}
\mbox{Argmin} \left\{\widetilde I(x):\,x\in \mathbb R^d\right\}\,=\,\partial \widetilde \Lambda (0),
\end{equation}
where, we recall, $\partial \widetilde \Lambda (0)$ denotes the subdifferential\footnote{The subdifferential of a convex function $f:\mathbb R^d\mapsto\mathbb R$ at a point $x\in \mathbb R^d$ is the set of all points $s\in \mathbb R^d$ such that, for all $y\in \mathbb R^d$, $f(y)\geq f(x)+s^\top (y-x)$~\cite{Urruty}.} of $\widetilde \Lambda$ at $\lambda=0$. We will show that $\partial \widetilde \Lambda (0)$ is a singleton and that it equals $\{\widetilde m\}$. To do this, note that, by our assumption, $\mathcal D_i=\mathbb R^d$ for each $i$. Thus, $\mathcal D_{\widetilde \Lambda}=\mathbb R^d$. In particular, $0\in \mathcal D_{\widetilde \Lambda}^\mathrm{o}$ and the claim follows by combining parts~\ref{Lambda-at-0} and~\ref{Lambda-is-smooth} of Lemma~\ref{lemma-properties-of-lmgf} and noting that $\mathbb E[\sum_{j=1}^N a_j Z_{j,t}]=\sum_{j=1}^N a_jm_j=\widetilde m$. We conclude that $\widetilde I(x)=0$ if and only $x=\widetilde m$.

We next use the previous conclusion together with convexity of $\widetilde I$ to show that, for any $\epsilon>0$,
\begin{equation}
\label{claim-2}
\inf_{x: \,\|x-\widetilde m\|\geq \epsilon} \widetilde I(x)>0.
\end{equation}
First, since $\widetilde I$ is convex and it achieves its minimum at $\widetilde m$, it must be that $\widetilde I$ is nondecreasing along any half-line that starts at $\widetilde m$. Hence, $\inf_{t\in[\epsilon,+\infty)} \widetilde I(\widetilde m +td)= \widetilde I(\widetilde m +\epsilon d)$, for any $d$. This in particular implies that $\inf_{x\in \mathbb R^d:\|x-\widetilde m\|\geq \epsilon} \widetilde I(x) = \inf_{x\in \mathbb R^d:\|x-\widetilde m\|= \epsilon}  \widetilde I(x)$. To prove the claim in~\eqref{claim-2}, we need to show that the preceding infimum is strictly
greater than zero. Since $\widetilde I$ is lower semi-continuous and the set under the infimum $\partial B_{\widetilde m}(\epsilon)$ is compact, it follows by Weierstrass theorem that $\widetilde I$ attains a minimum on $\partial B_{\widetilde m}(\epsilon)$; denote this minimum by $\widehat x_{\epsilon}$. Recalling now that $\widetilde I(x)>0$ for any $x\neq \widetilde m$, we conclude that it must be that $\widetilde I(\widehat x_{\epsilon})>0$. This concludes the proof of the claim in~\eqref{claim-2}.

Having~\eqref{claim-2}, it is now easy to complete the proof of the second part of Theorem~\ref{theorem-deterministic}. Fix $i\in V$. From the upper bound of LDP, proved in the first part, and~\eqref{claim-2}, we have that for any $\epsilon>0$:
\begin{equation}
\limsup_{t\rightarrow +\infty}\,\frac{1}{t}\,\log\,\mathbb P\left(\left\|X_{i,t}-\widetilde m\right\|\geq \epsilon\right)
\leq - C_{\epsilon}<0,
\end{equation}
where we denoted $C_{\epsilon}= \widetilde I(\widehat x_{\epsilon})$. The previous inequality implies that for any $\delta>0$ we can find a constant $K_{\delta}$ such that, for all $t$, $\mathbb P\left(\left\|X_{i,t}-\widetilde m\right\|\geq \epsilon\right)\leq K_{\delta}e^{-(C_{\epsilon}-\delta)t}.$ Choosing for each $\epsilon$, $\delta=C_{\epsilon}/2$, we obtain exponential convergence of $X_{i,t}$ to $\widetilde m$. By the first Borel-Cantelli lemma~\cite{Karr}, this in turn implies almost sure convergence of $X_{i,t}$.
\end{proof}
Let $G$ denote the induced graph of $A$, i.e., $G=(V,E)$ where $E=\left\{ (i,j): A_{ji}>0\right\}$, e.g.,~\cite{Lobel11}.
\begin{corollary}
\label{corollary-between-one-and-all-deterministic}
When $Z_{i,t}$ are i.i.d. (identical agents), it holds
\begin{equation}
\label{eq-bounds-deterministic}
I\leq \widetilde I\leq NI,
\end{equation}
where $I$ is the conjugate of an agent's log-moment generating function $\Lambda\equiv\Lambda_j$ and the inequalities in~\eqref{eq-bounds-deterministic} hold in the pointwise sense. Moreover, the lower bound in~\eqref{eq-bounds-deterministic} is attained whenever there exists a ``leader'' agent $i$ that satisfies $A_{ii}=1$ and for any $j$ there is a (directed) path from $i$ to $j$ in the induced graph of $A$. The upper bound is attained when $A$ is doubly stochastic with positive diagonals and the induced graph of $A$ is strongly connected.
\end{corollary}
\begin{proof}
When $Z_{i,t}$ are i.i.d.,
\vspace{-2mm}
\begin{equation}
\label{eq-tilde-Lambda-for-iid}
\widetilde \Lambda(\lambda)= \sum_{j=1}^N \Lambda(a_j\lambda).
\end{equation}
\vspace{-0.5mm}
By Lemma~\ref{simple-lemma} applied to $\alpha=a$ (note that $a$ is a stochastic vector), from eq.~\eqref{eq-tilde-Lambda-for-iid} we obtain
\begin{equation*}
\lambda^\top x - \Lambda(\lambda) \leq \lambda^\top x - \widetilde \Lambda (\lambda)\leq \lambda^\top x - N\Lambda\left(1/N\lambda\right).
\end{equation*}
Taking the supremum, the right-hand side inequality in~\eqref{eq-bounds-deterministic} follows by the following simple manipulations $\sup_{\lambda \in \mathbb R^d}\lambda^\top x - N\Lambda\left(1/N\lambda\right)= N (\sup_{{\lambda^\prime} \in \mathbb R^d} {\lambda^\prime}^\top x - \Lambda(\lambda^\prime))= N I(x).$
 The left-hand side inequality in~\eqref{eq-bounds-deterministic} is proven similarly. By Lemma~\ref{simple-lemma}, the log-moment generating in~\eqref{eq-tilde-Lambda-for-iid} is upper bounded by $\Lambda$, and, by similar calculations as in the above, we get $\widetilde I(x)= \sup_{\lambda\in \mathbb R^d} \lambda^\top  x- \widetilde \Lambda(\lambda)\geq \sup_{\lambda\in \mathbb R^d} \lambda^\top x - \Lambda(\lambda)=I(x).$

To prove the second part of Corollary~\ref{corollary-between-one-and-all-deterministic}, suppose that, for some $i$, $A_{ii}=1$ and that in the induced graph of $A$ there is a path from $i$ to any node $j$. By Theorem~\ref{theorem-deterministic} the claim is proven if we show that $\lambda_2(A)<1$ and that $e_i$ is the left eigenvector of $A$ corresponding to the eigenvalue $1$. Let $G$ denote the induced graph of $A$. To prove that $\lambda_2(A)<1$, it suffices to show that $G$ has exactly one initial class and that this class is aperiodic~\cite{Kirkland09}. Since $A$ is stochastic and $A_{ii}=1$, we have $A_{ij}=0$ for all $j\in V$, $j\neq i$. Thus, $\{i\}$ is an initial class of $G$. We next show that this is in fact the only initial class in $G$. Fix a node $j\neq i$ and let $C(j)$ denote the class of $G$ that $j$ belongs to. Note that $C(j)$ cannot contain $i$ (otherwise it would be possible to reach $i$ from $j$, which can't be true because $A_{il}=0$, for all $l\neq i$, and thus there are no edges pointing to $i$). Since (by our assumption) $j$ can be reached by a directed path from $i$, and, on the other hand, $i\notin C(j)$, there must be an edge pointing to $C(j)$. Hence, $C(j)$ is not an initial class of $G$. Repeating this for every $j$, we prove that there are no other initial class of $G$ beside $\{i\}$. Finally, it easy to see that $\{i\}$ is also aperiodic ($A_{ii}>0$), hence proving that $\lambda_2(A)<1$.

It only remains to verify that $e_i$ is the left eigenvector: since $A_{ii}=1$ and $A$ is stochastic, the $i$-th row of $A$, $e_i^\top A$, equals $e_i$. This completes the proof of the claim.

Suppose now that $A$ is doubly stochastic with positive diagonals and a strongly connected induced graph. Similarly as with the lower bound, by Theorem~\ref{theorem-deterministic}, it is sufficient to prove that $1/N 1_N$ is the left eigenvector of $A$ corresponding to the eigenvalue $1$ and that $\lambda_2(A)<1$. Since $A$ is doubly stochastic, it must be that $a^\top A=a^\top$ for $a=1/N 1_N$. Finally, since $A$ has positive diagonals and a strongly connected induced graph, we have that $A$ is irreducible and aperiodic, and hence $\lambda_2(A)<1$~\cite{Kirkland09} (see also Corollary 8.4.8. in~\cite{MatrixAnalysis}). This completes the proof of Corollary~\ref{corollary-between-one-and-all-deterministic}.
\end{proof}

\vspace{-0.5mm}
\mypar{Rate $\widetilde I$ for Gaussian observations} Of special interest is the case when observations $Z_{i,t}$ are all Gaussian. For this case, Lemma~\ref{lemma-Gaussian} gives a closed form expression for the rate function~$\widetilde I$.
\begin{lemma}
\label{lemma-Gaussian}
Suppose that $Z_{j,t}\sim \mathcal N\left(m_j,S_j\right)$, for $j\in V$, where $S_j$, for each $j$, is a positive definite matrix. Function $\widetilde I$ from Theorem~\ref{theorem-deterministic} is then given by
\begin{equation}
\label{eq-analytical-for-Gaussian}
\widetilde I(x)= \frac{1}{2}(x-\widetilde m)\widetilde S^{-1}(x-\widetilde m),
\end{equation}
where $\widetilde m= \sum_{j=1}^N a_j m_j$ and $\widetilde S=\sum_{j=1}^N a_j^2 S_j.$
 In particular, when $m_j\equiv m$ and $S_j\equiv S$, $\widetilde I(x)= 1 /(\sum_{j=1}^N a_j^2) I(x)$, where $I(x)$ is the nodes' individual rate function given in Example~\ref{ex-lmgf-Gauss}.
\end{lemma}
\begin{proof}
Fix $x\in \mathbb R^d$ and recall that the log-moment generating function of a Gaussian vector of mean $m$ and covariance $S$ is $\lambda \mapsto \lambda^\top m + 1/2 \lambda^\top S\lambda$. Then $\widetilde \Lambda(\lambda) = \sum_{j=1}^N a_j \lambda^\top m_j + a_j^2 \frac{1}{2} \lambda^\top S_j\lambda,$
and thus
\begin{equation}
\label{eq-proof-widetilde-I-Gaussian}
\widetilde I(x)= \sup_{\lambda \in \mathbb R^d} \lambda^\top x - \sum_{j=1}^N a_j \left(\lambda^\top m_j + a_j^2 \frac{1}{2} \lambda^\top S_j\lambda\right).
\end{equation}
Since the function under the supremum is (strictly) concave, we obtain the optimizer $\lambda^\star$ from the first order optimality condition \[x-\sum_{j=1}^N a_j m_j -\sum_{j=1}^N a_j^2 S_j \lambda=0.\]
 It follows that $\lambda^\star= \left(\sum_{j=1}^N a_j^2 S_j\right)^{-1} \left(x-\sum_{j=1}^N a_j m_j\right)$, which, when inserted in~\eqref{eq-proof-widetilde-I-Gaussian}, yields the identity~\eqref{eq-analytical-for-Gaussian}.
\end{proof}

\begin{remark} It is possible to determine $\widetilde I$ analytically even when matrices $S_j$, $j=1,...,N$, and vector $a$ are such that $\widetilde S$ is not invertible. It can be shown that the expression for $\widetilde I$ for this case is:
\begin{equation*}
\widetilde I(x)=\left\{
\begin{array}{ll}
(x-\widetilde m)^\top \widetilde S^\dagger(x-\widetilde m),&x \in \mathcal R (\widetilde S)\\
+\infty,& \mathrm{otherwise}
\end{array}\right..
\end{equation*}
\end{remark}


\vspace{-1mm}
\section{Network design}
\label{sec-network-design}
From Theorem~\ref{theorem-deterministic} and Corollary~\ref{corollary-between-one-and-all-deterministic}, it is clear that the performance of algorithm~\eqref{alg-D-INF} critically depends on the choice of the weight matrix $A$, and in particular, on its left eigenvector $a$. We therefore pose the problem of optimizing $a$, for a fixed desired accuracy region $C$:

\begin{equation}
\begin{array}[+]{ll}
\mbox{maximize} & \inf_{x\in \mathbb R^d\setminus C} \widetilde I(x) \\
\mbox{subject to} & a \in \Delta^{N-1}
\end{array},
\label{design-problem}
\end{equation}
where $\widetilde I$ is the rate function from Theorem~\ref{theorem-deterministic}. We denote by $a^\star_C$ and $I^{\star}_C$, respectively, an optimal solution and the optimal value of problem~\eqref{design-problem}.

We exploit the analytical expression~\eqref{eq-analytical-for-Gaussian} for the rate function from Lemma~\ref{lemma-Gaussian}, to show that, for the Gaussian observations, problem~\eqref{design-problem} can be solved efficiently. We assume that all the nodes are observing the same set of physical quantities $\theta=(\theta_1,...,\theta_d)^\top$, embedded in the local sensor noises. Hence, the observations $Z_{i,t}$ have the same expected value $\theta =:m \equiv m_i$ across all nodes. We show in Lemma~\ref{lemma-SDP} that when $C$ is a ball,~\eqref{design-problem} can be formulated as an SDP.
\begin{lemma}
\label{lemma-SDP}
Consider the setup of Lemma~\ref{lemma-Gaussian} when $m_i\equiv m$. When the confidence set $C$ is an Euclidean ball of some arbitrary radius $\zeta>0$ centered at $m$, $B_m(\zeta)$, the optimal solution of~\eqref{design-problem} is obtained by solving:
\begin{equation}
\begin{array}[+]{lc}
\mbox{minimize} &\gamma \\
\mbox{subject to} & \left[\begin{array}{ll}
\gamma I_d & \mathcal I\, \widetilde {\mathcal S}\\
\widetilde {\mathcal S}\, \mathcal I^\top & I_{Nd}
\end{array}\right]\succeq 0\\
& a \in \Delta^{N-1}
\end{array},
\label{design-via-SDP}
\end{equation}
where $\widetilde {\mathcal S} \in \mathbb R^{Nd\times Nd}$ is a block diagonal matrix given by $\widetilde {\mathcal S}=\mathrm{diag}\left\{a_1 S_1^{1/2},\ldots,a_N S_N^{1/2}\right\}$, and $\mathcal I=[I_d\,\ldots I_d]\in \mathbb R^{d\times Nd}$, where $I_d$ repeats $N$ times. Furthermore, $I^\star_C=\zeta^2/(2 \gamma^\star)$, where $\gamma^\star$ is the optimum of~\eqref{design-via-SDP}.
\end{lemma}
\begin{remark} Although problem~\eqref{design-problem} involves the expected value of the observations $m$ (which we don't know), it is clear from the equivalent reformulation~\eqref{design-via-SDP} that, under the stated assumption, the knowledge of $m$ is not needed for discovering the optimal $a$ in~\eqref{design-problem}. We also remark that, for the same assumptions, the solution of~\eqref{design-problem} does not depend on the particular accuracy $\zeta$: once~\eqref{design-via-SDP} is solved, the same vector $a^\star_C$ applies for all $C=B_m(\zeta)$, $\zeta>0$.
\end{remark}
\begin{remark}
When the observations are one-dimensional ($d=1$), it can be shown that the SDP in~\eqref{design-via-SDP} reduces to a quadratic program (QP).
\end{remark}
\begin{proof}
We start by finding a closed form expression for the objective function $\inf_{x\in \mathbb R^d\setminus B_m(\zeta)} \widetilde I(x)$, for a given vector $a$. Similarly as in the proof of Theorem~\ref{theorem-deterministic}, it can be shown that for any $\zeta>0$,
\begin{align*}
\inf_{x\in \mathbb R^d\setminus B_m(\zeta)} \widetilde I(x)  = \inf_{x\in \mathbb R^d: \|x-m\|\geq\zeta} \widetilde I(x) = \min_{x\in \mathbb R^d: \|x-m\|=\zeta} \widetilde I(x).
\end{align*}
It is easy to see that the latter problem can be reformulated as:
\begin{equation}
\min_{v \in \mathbb R^d: \|v\|=1} \frac{\zeta^2}{2} v^\top \widetilde S^{-1} v = 1/\lambda_{\max}(\widetilde S).
\end{equation}
Maximizing $1/\lambda_{\max}(\widetilde S)$ corresponds to minimizing $\lambda_{\max}(\widetilde S)$ and hence we obtain that~\eqref{design-problem} is equivalent to:
\begin{equation}
\begin{array}[+]{ll}
\mbox{minimize} & \lambda_{\max}\left( \sum_{j=1}^Na_j^2 S_j\right) \\
\mbox{subject to} &   a \in \Delta^{N-1}
\end{array}
\label{reformulation-1}
\end{equation}
where the optimal value of~\eqref{design-problem} $I^\star_C$ is obtained as $\zeta^2/(2 \lambda^\star)$, where $\lambda^\star$ is the optimal value of~\eqref{reformulation-1}.
We next show that~\eqref{reformulation-1} can be recast in the SDP form~\eqref{design-via-SDP}. Introducing the epigraph variable $\gamma \in \mathbb R$~\cite{BoydsBook} yields the constraint $\sum_{j=1}^N a_j^2 S_j\preceq \gamma I_d$, which can be equivalently represented as $\gamma I_d - \mathcal I \widetilde {\mathcal S} (I_{Nd})^{-1}\widetilde {\mathcal S} \mathcal I^\top\succeq 0$. Since the identity matrix $I_{Nd}$ is positive definite, equivalence of~\eqref{reformulation-1} and~\eqref{design-via-SDP} follows from the Schur complement theorem \cite{BoydsBook}.
\end{proof}

\section{Universal bounds on the rate functions for general, random weight matrices}
\label{sec-Generic}
We have seen in the previous section (Corollary~\ref{corollary-between-one-and-all-deterministic}) that, when the weight matrices $W_t$ are deterministic and constant, the states exhibit a very interesting and fundamental property: their large deviation probabilities $\mathbb P\left(X_{i,t}\in D\right)$ have the rates that are always lower than the corresponding rate of the fusion center, and always higher than the corresponding rate of a node working in isolation. Theorem~\ref{theorem-universal-limits} that we present next asserts that this property in fact holds, not only for deterministic, but for arbitrary sequences of random weight matrices.
\begin{theorem}
\label{theorem-universal-limits}
Consider the distributed inference algorithm~\eqref{alg-D-INF} under Assumptions~\ref{ass-network-and-observation-model} and~\ref{ass-finite-at-all-points}, when $Z_{i,t}$ are i.i.d. (identical agents). For any measurable set~$G\subseteq \mathbb R^d$, for each $i$:
\begin{align}
\label{eq-universal-bounds-upper}
- \inf_{x\in G^{\mathrm{o}}} N I(x) &\leq  \liminf_{t\rightarrow +\infty}\,\frac{1}{t}\,\log \mathbb P\left(X_{i,t}\in G\right)\\
\label{eq-universal-bounds-lower}
&\!\!\!\!\!\!\!\!\!\!\!\!\!\!\!\!\!\!\!\! \leq \limsup_{t\rightarrow +\infty}\,\frac{1}{t}\,\log \mathbb P\left(X_{i,t}\in G\right)  \leq - \inf_{x\in \overline G} I(x).
\end{align}
\end{theorem}
Theorem~\ref{theorem-universal-limits} asserts that, no matter how we design the agents' interactions (represented by the weight matrices), in terms of large deviations performance, algorithm~\eqref{alg-D-INF} can never be worse than when a node is working in isolation, but it also can never beat the fusion center. This result is important as it provides fundamental bounds for large deviations performance of \emph{any} algorithm of the form~\eqref{alg-D-INF} that satisfies Assumptions~\ref{ass-network-and-observation-model} and~\ref{ass-finite-at-all-points} and processes i.i.d. observations. 
In the next two subsections we state our proofs of Theorem~\ref{theorem-universal-limits}.

\subsection{Proof of the upper bound}
\label{sec-Upper-Proof}

Fix an arbitrary $i\in V$. To prove~\eqref{eq-universal-bounds-upper} for node $i$, it suffices to show that, for any closed set $F$,
\begin{equation}
\label{suffices-upper-bound}
\limsup_{t\rightarrow+\infty}\,\frac{1}{t}\log \mathbb P\left( X_{i,t}\in F\right)\leq -\inf_{x\in F} I(x).
\end{equation}
To see why this is true, note that, for an arbitrary measurable set $D$, there holds $\mathbb P\left(X_{i,t}\in D\right)\leq \mathbb P\left(X_{i,t} \in \overline D\right)$. Applying~\eqref{suffices-upper-bound} to the closed set $F=\overline D$ yields~\eqref{eq-universal-bounds-upper}.

The proof of~\eqref{suffices-upper-bound} consists of the following three steps.

\underline{\emph{Step 1:}} We use the exponential Markov inequality, together with conditioning on the matrices $W_1,...,W_t$, to show that, for any measurable set $D\subseteq \mathbb R^d$,
\begin{equation}
\label{eq-it-suf-ball}
\frac{1}{t}\log \mathbb P\left(X_{i,t} \in D\right) \leq - \inf_{x\in D} \lambda^\top x - \Lambda(\lambda).
\end{equation}

\underline{\emph{Step 2:}} In the second step, we show that~\eqref{eq-it-suf-ball} is a sufficient condition for~\eqref{suffices-upper-bound} to hold for all \emph{compact} sets $F$. Lemma~\ref{lemma-step-1} formalizes this statement.
\begin{lemma}
\label{lemma-step-1}
Suppose that~\eqref{eq-it-suf-ball} holds for any measurable set $D\subseteq \mathbb R^d$. Then the inequality~\eqref{suffices-upper-bound} holds for all compact sets~$F$.
\end{lemma}
The proof of Lemma~\ref{lemma-step-1} uses the standard ``finite cover'' argument: for a compact set $F$, a finite number of balls forming a cover of $F$ is constructed, and then~\eqref{eq-it-suf-ball} is applied to each of the balls. The details of this derivation are given in Appendix~\ref{app:finite-cover-compact-sets}.

\underline{\emph{Step 3:}} So far, \emph{Steps 1} and \emph{2} together imply that~\eqref{suffices-upper-bound} holds for all compact sets. To extend~\eqref{suffices-upper-bound} to all \emph{closed} sets $F$, by a well known result from large deviations theory, Lemma~1.2.18 from~\cite{DemboZeitouni}, it suffices to show that the sequence of measures $\mu_{i,t}: \mathcal B(\mathbb R^d)\mapsto [0,1]$, $\mu_{i,t}(D):= \mathbb P\left(X_{i,t}\in D\right)$ is exponentially tight. We prove this by considering the family of compact sets $H_{\rho}:=[-\rho,\rho]^d$, with $\rho$ increasing to infinity. The result is given in Lemma~\ref{lemma-exp-tightness}, and the proof can be found in Appendix~\ref{app:exponential-tightness}.
%
\begin{lemma}
\label{lemma-exp-tightness}
For every $i\in V$,
\begin{equation}
\lim_{\rho\rightarrow +\infty} \,\limsup_{t\rightarrow  +\infty} \,\mu_{i,t}\left( H_\rho^c\right) \leq -\infty.
\end{equation}
Hence, the sequence $\left\{\mu_{i,t}\right\}_{t=1,2,...}$ is exponentially tight.
\end{lemma}

We now provide the details of \emph{Step 1}.

\emph{Step 1.} The proof of~\eqref{eq-it-suf-ball} is based on two key arguments: exponential Markov inequality~\cite{Karr} and the right hand side inequality of Lemma~\ref{simple-lemma}. For any measurable set $D\subseteq \mathbb R^d$ and any $\lambda \in \mathbb R^d$, by the exponential Markov inequality, we have
\begin{equation}
\label{eq-Markov-bound}
1_{\left\{X_{i,t} \in D\right\}}\leq e^{t \lambda^\top X_{i,t} - t \inf_{x\in D} \lambda^\top x},
\end{equation}
which, after computing the expectation, yields
\begin{equation}
\label{eq-Markov-bound-exp}
\mathbb P\left(X_{i,t} \in D\right) \leq e^{- t \inf_{x\in D} \lambda^\top x} \mathbb E\left[e^{t \lambda^\top X_{i,t}}\right] .
\end{equation}
We now focus on the right hand side of~\eqref{eq-Markov-bound-exp}. Conditioning on $W_1,\ldots,W_t$, the summands in~\eqref{alg-D-INF-compact} become independent, and using the fact that the $Z_{i,t}$'s are i.i.d. with the same log-moment generating function $\Lambda$, we obtain
\begin{align}
\label{eq-independence}
\mathbb E\left[\left.e^{t \lambda^\top X_{i,t}}\right| W_1,...,W_t\right]
= e^{\sum_{s=1}^t \sum_{j=1}^N \Lambda \left([\Phi(t,s)]_{ij} \lambda\right)}.
\end{align}
Applying now the right-hand side inequality Lemma~\ref{simple-lemma} to $\sum_{j=1}^N \Lambda \left([\Phi(t,s)]_{ij} \lambda\right)$ for each fixed $s$ (note that, for a fixed $s$, $\left[[\Phi(t,s)]_{i1},...,[\Phi(t,s)]_{iN}\right]\in \Delta^{N-1}$ ), it follows that the conditional expectation above is upper bounded by $e^{t\Lambda(\lambda)}$, i.e.,
\begin{equation}
\label{eq-independence-2}
\mathbb E\left[\left.e^{t \lambda^\top X_{i,t}}\right| W_1,...,W_t\right] \leq e^{t \Lambda(\lambda)}.
\end{equation}
for any $\lambda\in \mathbb R^d$. Since in~\eqref{eq-independence-2} $W_1,...,W_t$ were arbitrary, taking the expectation, we get $\mathbb E\left[e^{t \lambda^\top X_{i,t}}\right] \leq e^{t \Lambda(\lambda)}.$
Combining this with~\eqref{eq-Markov-bound-exp}, we finally obtain
\begin{equation}
\label{eq-Markov-bound-exp-2}
\frac{1}{t}\log \mathbb P\left(X_{i,t} \in D\right) \leq - \inf_{x\in D} \lambda^\top x + \Lambda(\lambda).
\end{equation}

\subsection{Proof of the lower bound}
\label{sec-Lower-Proof}
We prove~\eqref{eq-universal-bounds-lower} following the general lines of the proof of the G\"{a}rtner-Ellis theorem lower bound, see~\cite{DemboZeitouni}. However, as we will see later in this proof, we encounter several difficulties along the way
 that force us to depart from the standard G\"{a}rtner-Ellis method and use finer arguments.
 The main reason for this is that, in contrast with the setup of the G\"{a}rtner-Ellis theorem, the sequence of the (scaled) log-moment generating functions of $X_{i,t}$ (see ahead~\eqref{def-Lambda-t}) need not have a limit. Nevertheless, with the help of Lemma~\ref{simple-lemma}, we will be able to ``sandwich'' each member of this sequence between $\Lambda(\cdot)$ and $N \Lambda\left(1/N \cdot\right)$.
  This is the key ingredient that allows us to derive~\eqref{eq-universal-bounds-lower}. The proof is organized in the following four steps.

\mypar{\underline{\emph{Step 1}}} In this step, we derive a sufficient condition, given in Lemma~\ref{lemma-x-in-D^o-conditioning}, for~\eqref{eq-universal-bounds-lower} to hold. Namely, to prove~\eqref{eq-universal-bounds-lower} for a given set $D$, it suffices to confine $X_{i,t}$ to a smaller region $B_x(\delta)$ within $D$, and show that, conditioned on any realization of the matrices $W_1,...,W_t$, the rate of this event is at most $NI(x)$. Lemma~\ref{lemma-x-in-D^o-conditioning} is proven by applying Fatou's lemma~\cite{Karr} to the sequence of random variables $R_t:=\frac{1}{t}\log \mathbb P\left(X_{i,t}\in D|W_1,...,W_t\right)$, and then combining the obtained result with the simple fact that, for every $x\in D^{\mathrm{o}}$ and all $\delta$ sufficiently small, $B_x(\delta)\subseteq D$. The proof is given in Appendix~\ref{app:proof-of-lemma-x-in-D^o-conditioning}.
\begin{lemma}
\label{lemma-x-in-D^o-conditioning}
If for every $x\in \mathbb R^d$ and $\omega \in \Omega$,
\begin{equation}
\label{eq-it-suffices-to-consider-delta-balls}
\!\lim_{\delta\rightarrow 0}\liminf_{t\rightarrow +\infty}\frac{1}{t}\log \mathbb P\left( X_{i,t}\in B_x(\delta)| W_1,...,W_t\right)\! \geq\!\! - N I(x),\!\!\!\!\!\!\!\!\!
\end{equation}
then~\eqref{eq-universal-bounds-lower} holds for all measurable sets $D$.
\end{lemma}
\mypar{\underline{\emph{Step 2}}} To prove~\eqref{eq-it-suffices-to-consider-delta-balls}, we introduce the scaled log-moment generating function of $X_{i,t}$, under the conditioning on $W_1,...,W_t$,
\begin{equation}
\label{def-Lambda-t}
\Lambda_{t}(\lambda):= \frac{1}{t}\log \mathbb E\left[ \left. e^{t\lambda^\top X_{i,t}}\right|W_1,\ldots,W_t\right].
\end{equation}
It can be shown (similarly as in \emph{Step 1} of the proof of the upper bound) that, for any $\lambda\in \mathbb R^d$,
\begin{equation}
\label{eq-Lambda-t-equals}
\Lambda_{t}(\lambda)= \frac{1}{t} \sum_{s=1}^t\sum_{j=1}^N \Lambda\left( [\Phi(t,s)]_{ij} \lambda\right),
\end{equation}
where $\Phi(t,s)=W_t\cdots W_s$. Note that $\Lambda_t$ is convex and differentiable. However, $\Lambda_t$ is not necessarily $1$-coercive~\cite{Urruty}, which is needed to show~\eqref{eq-it-suffices-to-consider-delta-balls} for all points\footnote{More precisely, the problem arises when $x$ is not an exposed point of the conjugate $I_t$ of $\Lambda_t$, as will be clear from later parts of the proof (see also Exercise~{2.3.20} in~\cite{DemboZeitouni}).} $x\in \mathbb R^d$. To overcome this, we introduce a small Gaussian noise to the states $X_{i,t} $ and define, for each $t$, $Y_{i,t}=X_{i,t}+V/{\sqrt {M t}}$, where $V$ has the standard multivariate Gaussian distribution $\mathcal N(0_d,I_d)$, and, we assume, is independent of $Z_{j,t}$ and $W_t$, for all $j$ and $t$ (hence, $V$ is independent of $X_{i,t}$, for all $t$). The parameter $M>0$ controls the magnitude of the noise, and the factor $1/{\sqrt t}$ adjusts the noise variance to the same level of the variance of $X_{i,t}$.

For each fixed $M$, let $\Lambda_{t,M}$ denote the log-moment generating function associated with the corresponding $Y_{i,t}$, under the conditioning on $W_1,...,W_t$. It can be shown, using the independence of $V$ and $X_{i,t}$, that
\begin{align}
\label{eq-Lambda-t-m-equals}
\Lambda_{t,M} (\lambda) = \Lambda_t(\lambda) +  \frac{\|\lambda\|^2}{2M}, \;\;\lambda \in {\mathbb R}^d.
\end{align}
Hence, the noise adds a (strictly) quadratic function to $\Lambda_t$, thus making $\Lambda_{t,M}$ $1$-coercive, as proved in the following lemma.  Lemma~\ref{lemma-Lambda-t-M-is-nice} gives the properties of $\Lambda_{t,M}$ that we use in the sequel; the proof is given in Appendix~\ref{app:proof-of-lemma-Lambda-t-M}.
\begin{lemma}
\label{lemma-Lambda-t-M-is-nice}
\begin{enumerate}
\item \label{lemma-Lambda-t-M-part-coercive}
Function $\Lambda_{t,M}$ is convex, differentiable, and $1$-coercive. Thus, for any $x\in \mathbb R^d$, there exists $\eta_t=\eta_t(x)$ such that $\nabla \Lambda_{t,M}(\eta_t)=x$.
\item \label{lemma-Lambda-t-M-part-unifbdd}
Let $\theta=\mathbb E\left[Z_{i,t}\right]$. For any $x$, the corresponding sequence $\eta_t$, $t=1,2,...$, is uniformly bounded, i.e.,
\begin{equation}
\left\|\eta_t\right\|\leq M\left\| x-\theta\right\|, \mbox{\;for all }t.
\end{equation}
\end{enumerate}
\end{lemma}
Using the results of Lemma~\ref{lemma-Lambda-t-M-is-nice}, we prove in $\emph{Step 3}$ the counterpart of~\eqref{eq-it-suffices-to-consider-delta-balls} for the sequence $Y_{i,t}$ -- \eqref{eq-delta-ball-for-regularized}, and in $\emph{Step 4}$ we complete the proof of~\eqref{eq-universal-bounds-lower} by showing that~\eqref{eq-it-suffices-to-consider-delta-balls} (a sufficient condition for~\eqref{eq-universal-bounds-lower}) is implied by~\eqref{eq-delta-ball-for-regularized}.

\mypar{\underline{\emph{Step 3}}} We show that, for any fixed $x$, $M>0$, and $\omega\in \Omega$,
\begin{equation}
\label{eq-delta-ball-for-regularized}
\lim_{\delta\rightarrow 0}\liminf_{t\rightarrow +\infty}\frac{1}{t}\log \nu_{t,M}\left( B_{x}(\delta)\right) \geq - N I(x).
\end{equation}
where $\nu_{t,M}$ is the conditional probability measure induced by $Y_{i,t}$, $\nu_{t,M}(D)=\mathbb P\left( Y_{i,t}\in D|W_1,...,W_t\right)$, $D\in \mathcal B(\mathbb R^d)$.

To this end, fix an arbitrary $x, \delta,M,$ and $\omega$. We prove~\eqref{eq-delta-ball-for-regularized} by the change of measure argument. For any $t\geq 1$, we use the point $\eta_t$ from Lemma~\ref{lemma-Lambda-t-M-is-nice} to change the measure on $\mathbb R^d$ from $\nu_{t,M}$ to $\widetilde \nu_{t,M}$ by:
\begin{equation}
\label{eq-change-of-measure}
\frac{d\widetilde \nu_{t,M}}{d \nu_{t,M}}(z)= e^{\,t\, \eta_t^\top z \,- \,t\, \Lambda_{t,M}(\eta_t)},\:\:z \in {\mathbb R}^d.
\end{equation}
Note that, in contrast with the standard method of G\"{a}rtner-Ellis Theorem where the change of measure is fixed (once $x$ is given),
here we have a different change of measure\footnote{The reason for this alteration of the standard method
is the fact that our sequence of functions $\Lambda_{t,M}$ does not have a limit.}
\footnote{It can be shown that all distributions $\widetilde \nu_{t,M}$, $t\geq 1$, have the same expected value $x$;
we do not pursue this result here, as it is not crucial for our goals.}  for each $t$.
%
Expressing the probability $\nu_{t,M}\left(B_x(\delta)\right)$ through $\widetilde \nu_{t,M}$, for each $t$, we get:
\begin{align}
\label{eq-step-1}
&\frac{1}{t}\log \nu_{t,M}\left(B_x(\delta)\right)=\nonumber\\
& \phantom{=}= \Lambda_{t,M}(\eta_t) - \eta_t^\top x
+ \frac{1}{t} \log \int_{z\in B_x(\delta)} e^{t \eta_t^\top (x-z)} d\widetilde \nu_{t,M}(z)\nonumber \\
&\phantom{=}\geq  \Lambda_{t,M}(\eta_t) - \eta_t^\top x - \delta \left\| \eta_t\right\|
+ \frac{1}{t} \log \widetilde \nu_{t,M}\left( B_x(\delta)\right).
\end{align}
We analyze separately each of the terms in~\eqref{eq-step-1}. First, since $\eta_t$ is uniformly bounded, by Lemma~\ref{lemma-Lambda-t-M-is-nice}, we immediately obtain that the third term vanishes:
\begin{equation}
\label{eq-liminf-eta-t}
\lim_{\delta\rightarrow +0} \liminf_{t\rightarrow +\infty} - \delta \left\| \eta_t\right\|  \geq - \lim_{\delta \rightarrow 0}\delta M\|x-\theta\|=0.
\end{equation}
We consider next the sum of the first two terms. Let $I_{t,M}$ denote the conjugate of $\Lambda_{t,M}$. By Lemma~\ref{lemma-Lambda-t-M-is-nice}, we have that $\eta_t$ is the maximizer of $\lambda \mapsto  \lambda^\top x - \Lambda_{t,M}(\lambda)$. Thus, the sum of the first two terms in~\eqref{eq-step-1} equals $- I_{t,M}(x)=\Lambda_{t,M}(\eta_t)- \eta_t^\top x $. Further, starting from the fact that $\Lambda_{t,M}\geq \Lambda_t$ and then invoking Lemma~\ref{simple-lemma} (lower bound), we obtain:
\begin{align}
\label{eq-I-t-M-less-than-NI}
I_{t,M}(x) \leq \sup_{\lambda \in \mathbb R^d} \lambda^\top x- \Lambda_t(\lambda) &\leq \sup_{\lambda \in \mathbb R^d} \lambda^\top x- N\Lambda(\lambda/N)\nonumber\\
& =NI(x),
\end{align}
which holds for all $t\geq 1$ and all $M>0$. Comparing with~\eqref{eq-delta-ball-for-regularized}, we see that it only remains to show that the lim inf as $t\rightarrow +\infty$ of the last term in~\eqref{eq-step-1} vanishes with $\delta$.

It is easy to show that the log-moment generating function associated with $\widetilde \nu_{t,M}$ is $\widetilde \Lambda_{t,M}:=\Lambda_{t,M}(\lambda+\eta_t) - \Lambda_{t,M}(\eta_t)$. Let $\widetilde I_{t,M}$ denote the conjugate of $\widetilde \Lambda_{t,M}$. Similarly as in the proof of the upper bound in Section~\ref{sec-Upper-Proof}, it can be shown that
\begin{equation}
\label{eq-if-we-manage-we-prove}
\frac{1}{t}\log \widetilde \nu_{t,M}\left(B^{\mathrm c}_x(\delta)\right)\leq -\inf_{w\in B^{\mathrm c}_x(\delta)} \widetilde I_{t,M}(w).
\end{equation}
The next lemma asserts that the right-hand side of~\eqref{eq-if-we-manage-we-prove} is strictly negative\footnote{In the proof of the lower bound of the G\"{a}rtner-Ellis theorem, the sequence $\widetilde \Lambda_{t}$ (our $\widetilde \Lambda_{t,M}$) has a limit $\widetilde \Lambda$, and, because of this, it is sufficient to show that $\inf_{w\in B^{\mathrm c}_x(\delta)} \widetilde I(w)$ is strictly negative, where $\widetilde I$ is the conjugate of $\widetilde \Lambda$. Here, since we do not have the limit of the $\widetilde \Lambda_{t,M}$s, we need to prove that the latter holds for each function of the sequence $\widetilde I_{t,M}$, $t\geq 1$, and moreover, that the strict negativity does not ``fade out'' with $t$.}, and uniformly bounded away from zero. The proof is given in Appendix~\ref{app:proof-xi}.
\begin{lemma}
\label{lemma-widetilde-I-t-M}
For any $t$, there exists a minimizer $w_t=w_t(x,\delta)$ of the optimization problem $\inf_{w\in B^{\mathrm c}_x(\delta)} \widetilde I_{t,M}(w)$. Furthermore, there exists $\xi= \xi(x,\delta)>0$ such that
\begin{equation}
\label{eq-widetilde-I-t-M}
\widetilde I_{t,M} \left(w_t\right) \geq \xi, \mbox{\;for all }t.
\end{equation}
\end{lemma}
Combining~\eqref{eq-if-we-manage-we-prove} and~\eqref{eq-widetilde-I-t-M}, we get
\begin{equation*}
\widetilde \nu_{t,M}\left(B_x(\delta)\right) \geq 1- e^{-\xi t}, \mbox{\;for all }t.
\end{equation*}
which, together with the fact that $\widetilde \nu_{t,M}$ is a probability measure (and hence $\widetilde \nu_{t,M}\left(B_x(\delta)\right)\leq 1$), yields
\begin{equation}
\label{eq-limit-widetilde-nu-t}
\lim_{t\rightarrow +\infty} \frac{1}{t}\log \widetilde \nu_{t,M}\left(B_x(\delta)\right) = 0.
\end{equation}
Since~\eqref{eq-limit-widetilde-nu-t} holds for all $\delta>0$, we conclude that the last term in~\eqref{eq-step-1} vanishes after the appropriate limits have been taken. Summarizing~\eqref{eq-liminf-eta-t},~\eqref{eq-I-t-M-less-than-NI}, and~\eqref{eq-limit-widetilde-nu-t} finally proves~\eqref{eq-delta-ball-for-regularized}.

\mypar{\underline{\emph{Step 4}}} To complete the proof of~\eqref{eq-universal-bounds-lower}, it only remains to show that~\eqref{eq-delta-ball-for-regularized} implies~\eqref{eq-it-suffices-to-consider-delta-balls}. Since $X_{i,t}=Y_{i,t}- V/\sqrt{t M}$, we have
\begin{align}
\label{eq-smart}
& \mathbb P\left( X_{i,t} \in B_x(2\delta)|W_1,...,W_t\right) \nonumber\\
& \geq  \mathbb P\left( Y_{i,t} \in B_x(\delta),V/{\sqrt {t M}} \in B_x(\delta) |W_1,...,W_t\right)\nonumber\\
& \geq  \nu_{t,M}\left(B_x(\delta)\right) - \mathbb P\left( V/ \sqrt{tM} \notin B_x(\delta) \right).
\end{align}
From~\eqref{eq-delta-ball-for-regularized}, the rate for the probability of the first term in~\eqref{eq-smart} is at most $N I(x)$. On the other hand, the probability that the norm of $V$ is greater than $\sqrt {tM} \delta$ decays
exponentially with $t$ at the rate $M \delta^2/2$,
 \begin{equation}
 \label{eq-Gauss-term}
\lim_{t\rightarrow +\infty}\frac{1}{t}\log \mathbb P\left( V/\sqrt {tM} \in B_x(\delta) \right)= - \frac{M \delta^2}{2}.
 \end{equation}
Observe now that, for any fixed $\delta$, for all $M$ large enough so that $NI(x) < \frac{M \delta^2}{2}$,
the exponential decay of the difference in~\eqref{eq-smart} is determined by the rate of the first term, $N I(x)$.
This finally establishes~\eqref{eq-it-suffices-to-consider-delta-balls}, which combined with~\ref{lemma-x-in-D^o-conditioning} proves~\eqref{eq-universal-bounds-lower}.

\section{Simulation results}
\label{sec-Simul}

\begin{figure*}[t]
\vspace{-0.5cm}
\begin{center}$
\begin{array}{lll}
\adjustbox{trim={.085\width} {.2\height} {0.12\width} {.22\height},clip}%
  {\includegraphics[width=0.4\linewidth]{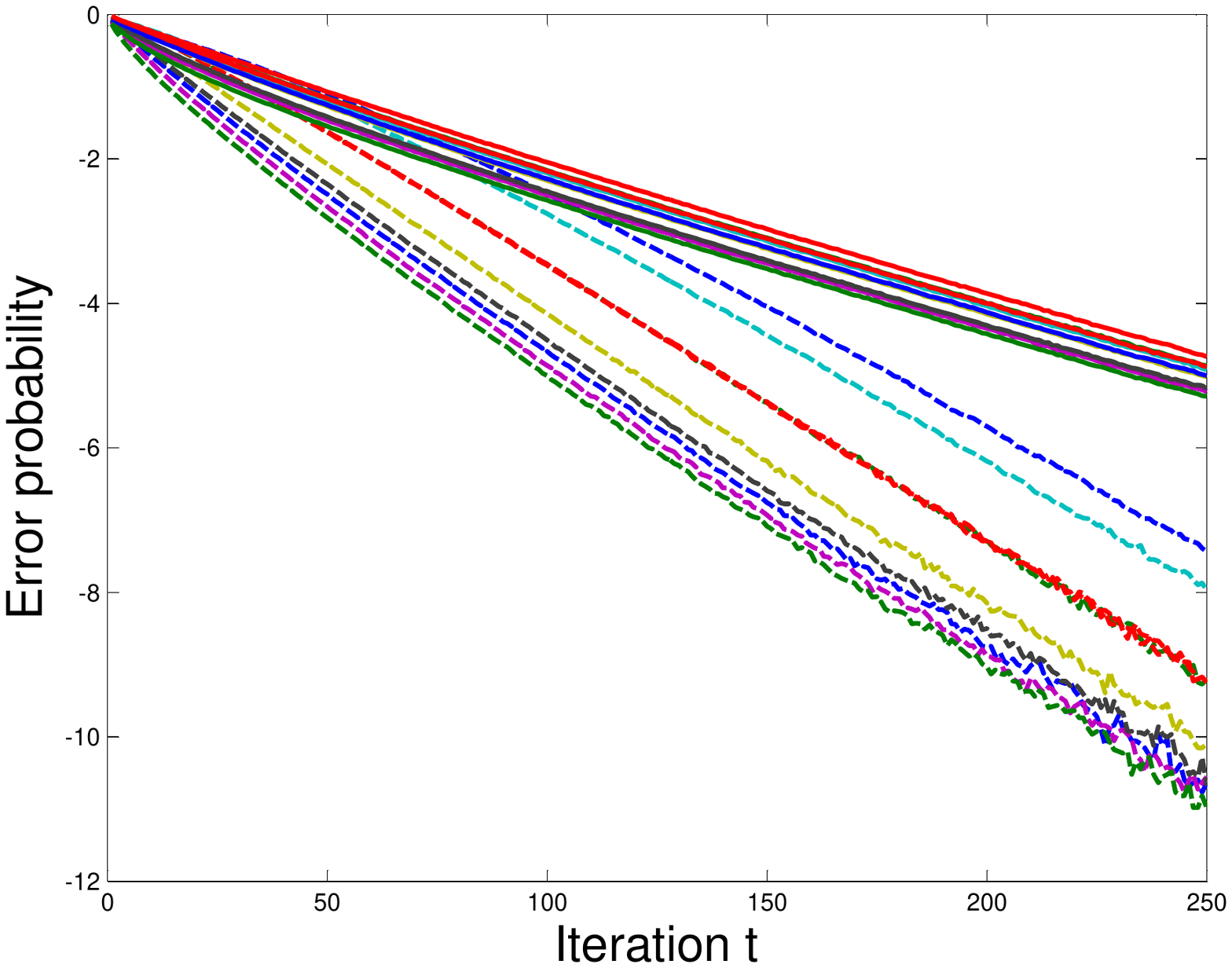}}&
  \adjustbox{trim={.04\width} {0\height} {0.14\width} {0\height},clip}%
{\includegraphics[width=0.39\linewidth]{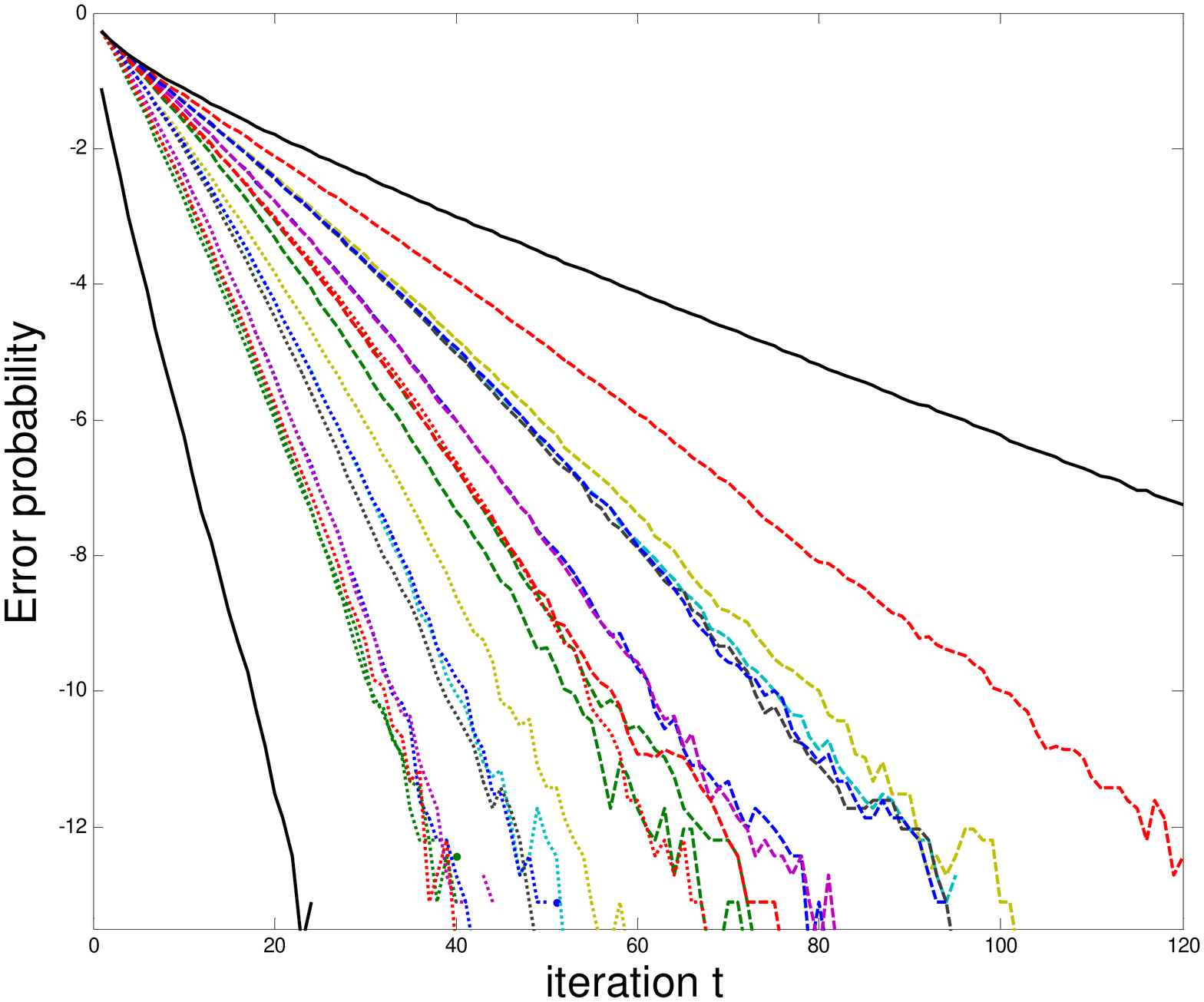}}&
\adjustbox{trim={.055\width} {0\height} {0.03\width} {0\height},clip}%
{\includegraphics[width=0.4\linewidth]{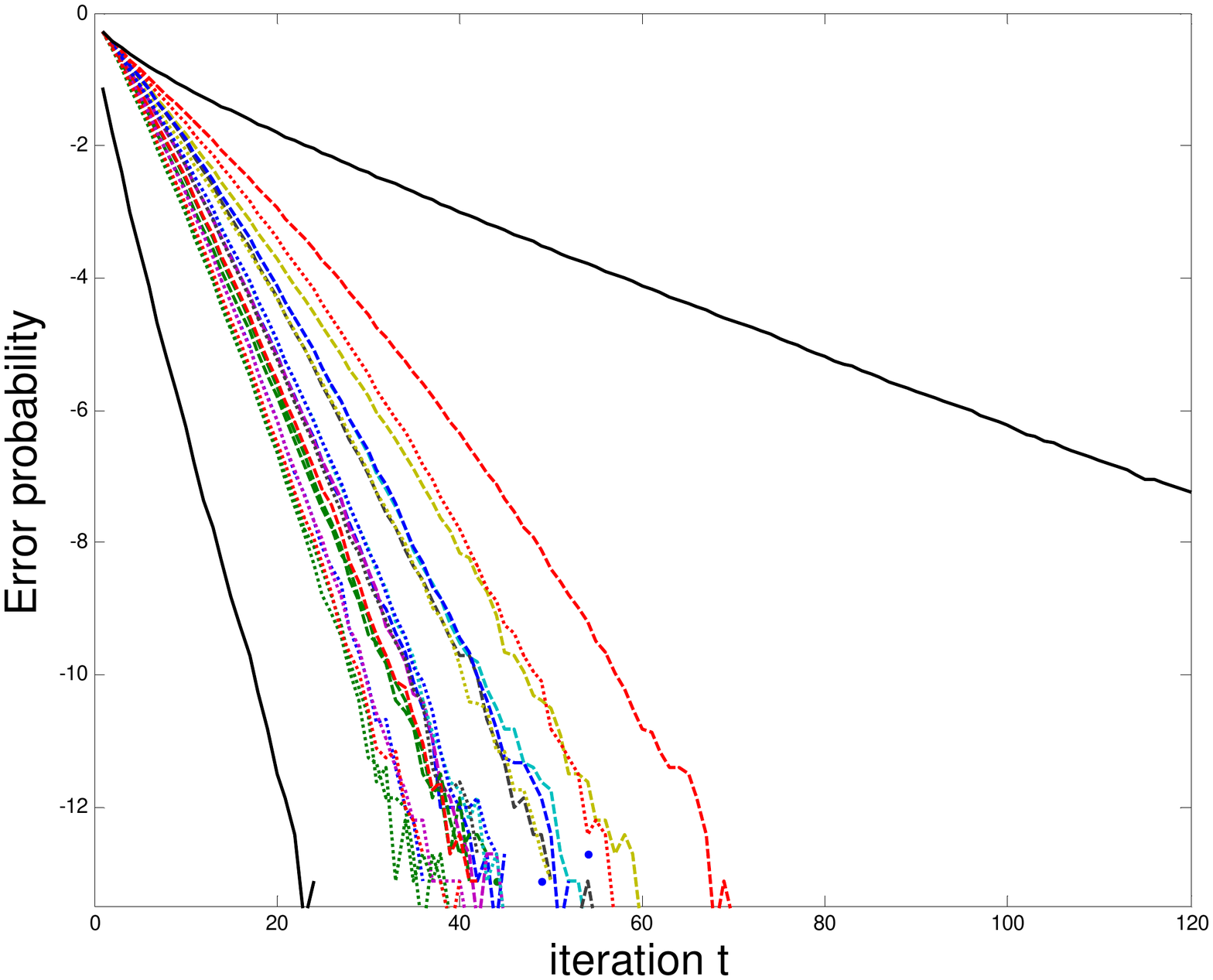}}
\end{array}$
\end{center}
\vspace{-0.3cm}
\caption{Estimated error probabilities $\widehat P_{i,t}$ vs. number of iterations $t$, for each $i$. Left (deterministic model): dotted lines correspond to $W_{\mathrm{opt}}$, and full lines correspond to $W_{\mathrm{unif}}$. Middle and Right (random model): dashed curves correspond to the i.i.d. model, dotted curves to the Markov chain model, and full curves to an isolated node (upper) and the fusion center (lower); $p=0.1$, $q_1=q_2=0.3$ (middle) and $p=0.5$, $q_1=0.7$, $q_2=0.1$ (right).}
  \label{Fig-deterministic-and-random}
\end{figure*}

This section presents our simulation results for the performance of algorithm~\eqref{alg-D-INF} for both deterministic and random weight matrices. In the deterministic case, we optimize the weight matrix $A$ by optimizing first its left eigenvector $a$; then we subsequently optimize $A$ such that it achieves the fastest averaging speed (see ahead~\eqref{design-of-W}), subject to the condition on the obtained left eigenvector. We estimate by Monte Carlo simulations the corresponding estimation error. Simulations show that the optimized system significantly outperforms the system where the eigenvector $a$ is uniform and $A$ is doubly stochastic, hence proving the benefit of network design. We then consider randomly time-varying weight matrices and verify by simulations Theorem~\ref{theorem-universal-limits} for the following cases: 1) $W_t$ are i.i.d. in time, with i.i.d. link failures; and 2) link failures of each link in the network, independently from other links, are governed by a Markov chain.

\mypar{Simulation setup} The number of nodes is $N=10$. Communication graph $\widehat G$
 is formed by placing the nodes uniformly at random in a unit square and forming the (biderectional) links between the nodes that are within distance $r=0.4$ from each other. Observations $Z_{i,t}$ are standard Gaussian, for each $i$, with the same expected value $m_i\equiv m$ chosen uniformly at random from the $[0,1]$ interval. In the deterministic case, the variances $S_i=\sigma_i^2$ are different across nodes, whereas in the random case all the nodes have the same variance $S=\sigma^2$. The quantities $S_i$, $i=1,...,N$, and $S$ are chosen uniformly at random, and, in the deterministic case, independently for each $i$, in the $[0,1]$ interval.

\subsection{Network design for the deterministic case}
\label{subsec-simul-design}
In this section, we consider the deterministic case, when the weight matrix is constant at all times, and when the observations are scalar ($d=1$). Since all the nodes have the same expected value $m$, we have that $\widetilde m=m$ (see Theorem~\ref{theorem-deterministic}), and thus all the states $X_{i,t}$ converge (a.s.) to $m$. We wish to find the weight matrix~$A=A_{\mathrm{opt}}$ that achieves this convergence with the fastest possible rate function $\widetilde I$. The accuracy region that we target is $C=[m- \zeta, m+\zeta]$, where we set $\zeta=0.035$.

We obtain $A_{\mathrm{opt}}$ as follows. We first solve problem~\eqref{design-via-SDP} via CVX~\cite{cvx},\cite{gb08} to obtain the optimal left eigenvector $a_{\mathrm{opt}}$ of $A_{\mathrm{opt}}$. Then, we optimize $A$ by minimizing the spectral norm of $A-1_N a_{\mathrm{opt}}^\top$, while respecting the sparsity pattern dictated by the communication graph $\widehat G$, as in~\cite{Boyd-fastavg04}. Hence, $A_{\mathrm{opt}}$ is obtained as the solution of the following optimization problem:
\begin{equation}
\begin{array}[+]{ll}
\mbox{minimize} &\left\| A-1_N a_{\mathrm{opt}}^\top \right\|\\
\mbox{subject to} & A 1_N=1_N\\
& a_{\mathrm{opt}}^\top A =a_{\mathrm{opt}}^\top\\
& A\in  \mathcal A
\end{array},
\label{design-of-W}
\end{equation}
where $\mathcal A:=\left\{A\in \mathbb R_{+}^{N\times N}: A_{ij}= 0,\mbox{\;if\;}(i,j)\notin \widehat G,\: i,j\in V\right\}$; see also Section~7.3 in~\cite{Boyd-fastavg04}. Note that the rate function is dependent on $A$ only through its left eigenvector $a$, but the significance of $\|A-1_N a^\top\|$  is in the finite time performance (i.e., vertical shift of the curves in Figure~\ref{Fig-deterministic-and-random} (bottom) further ahead). For the purpose of comparison, we also solve problem~\eqref{design-of-W} when $a_{\mathrm{opt}}$ is replaced by $a_{\mathrm{unif}}=1/N 1_N$; we denote the corresponding solution by $A_{\mathrm{unif}}$ ($A_{\mathrm{unif}}$ hence represents the doubly stochastic matrix with the fastest averaging on the same topology $\widehat G$ as $A_{\mathrm{opt}}$).

At each node~$i$ and each time $t$, we estimate the probability of error $\widehat P_{i,t}$, by Monte Carlo simulations: we count the number of times that the state of node $i$ at time $t$, $X_{i,t}$, falls outside of the accuracy region $C$, and then we divide this number by the number of simulation runs $K=1000000$, $\widehat P_{i,t}=\frac{1}{K}\sum_{k=1}^K\,1\{X^k_{i,t}\notin C\}$. We do this both for the case when algorithm~\eqref{alg-D-INF} runs with the weight matrix $A_{\mathrm{opt}}$ and when it runs with the weight matrix $A_{\mathrm{unif}}$.

The leftmost plot in Figure~\ref{Fig-deterministic-and-random} plots the evolution of the error probability over iterations, in the log-scale, for each node $i$; dotted lines correspond to $A_{\mathrm{opt}}$ while full lines correspond to $A_{\mathrm{unif}}$.
We can see from Figure~\ref{Fig-deterministic-and-random} (left) that for both $A_{\mathrm{opt}}$ and $A_{\mathrm{unif}}$ the curves at all nodes have the same slope, equal to the value of the corresponding rate function over the set $C$. For the same weight matrix, the vertical shift in different curves (that correspond to different nodes) is due to the difference in the observations parameters (intuitively, nodes with higher variances $\sigma^2_i$ need more time to filter out the noise -- and thus their error probability curves are shifted upwards), and the placement in the network (nodes with more central location in the network converge faster). We can see that the algorithm with the optimized left eigenvector achieves much higher large deviations rate than the one with the uniform eigenvector, as predicted by our theory. For example, for the target error probability of $e^{-5}\approx 0.007$, the optimized system requires around $140$ iterations on average (across nodes), while the system with the uniform vector $a$ needs around $250$ iterations for the same accuracy. The reason for this behavior is quite intuitive: optimizing the vector $a$ corresponds to choosing different weights for different sensors depending on their local variances (i.e., covariance matrices, when $d>1$).

%

\subsection{Random weight matrices}
\label{subsec-simul-random}
This subsection considers random weight matrices $W_t$ for two cases: i.i.d. link failures and Markov chain link failures. With the i.i.d. model, each directed link $(i,j)\in \widehat E$ can fail with probability $1-p$ at any given time $t$; this occurs independently from other link failures and independently from past times.
With the Markov chain model, each link $(i,j)\in \widehat G$ behaves as a Markov chain, independent from the Markov chains of other links, such that with probability $q_1$ the link stays online, if it was online in the previous time slot, and with probability $q_2$ stays offline. (For example, if at time $t$ a link is online, then at time $t+1$ this link stays online with probability $q_1$ and fails with probability $1-q_1$). 
With both i.i.d. and the Markov chain model, the weight matrix at time $t$ equals $W_t= I_N - \alpha L_t$, where $L_t$ is the Laplacian of the (directed) topology realization at time $t$, $\alpha=1/(d_{\max}+1)$, and $d_{\max}$ is the maximal degree in $\widehat G$.

The middle and the right plot in Figure~\ref{Fig-deterministic-and-random} show the estimated error probabilities versus the number of iterations for both the i.i.d. and the Markov chain model, for two different sets of parameters: $p=0.1$, $q_1=q_2=0.3$ (left) and $p=0.5$, $q_1=0.7$, $q_2=0.1$ (right). Both simulations are obtained for the same value of accuracy $\zeta=0.1$, and one-dimensional Gaussian observations with parameters $m$ and $S=\sigma^2$  chosen uniformly at random from the $[0,1]$ interval. The results for the i.i.d. model are plotted in dashed lines, while the results for the Markov chain model are plotted in dotted lines. For reference, we also plot the estimated error probabilities for perfect fusion and isolation (full lines), see Section~\ref{subsec-isolation-and-fusion}; the lower curve corresponds to fusion.
We can see from the plots that, under both models, the rate at which the error probability at each node decays is between the decay rate of the isolated node and fusion center curves, as predicted by Theorem~\ref{theorem-universal-limits}. We can also see that the agents' decay rates for the Markov chain model are faster than the ones for the i.i.d. model. This is expected since, for both sets of parameters, links in the i.i.d. model are online less frequently than the links in the Markov chain model, once the system reaches a stationary regime. Also, we see that improvements in the system parameters (higher $p$, in the i.i.d. model, and higher $q_1$ and lower $q_2$ in the Markov chain model) significantly affect the large deviations rates: in the right plot, the rates at each node got closer to the optimal, fusion center rate.


\section{Conclusion}
\label{sec-Concl}
We studied large deviations rates for consensus based distributed inference, for deterministic and random asymmetric weight matrices. For the deterministic case, we characterized the corresponding large deviations rate function, and we showed that it depends on the weight matrix only through its left eigenvector that corresponds to its unit eigenvalue. When the observations are Gaussian (not necessarily identically distributed across agents), the rate function has a closed form expression. Motivated by these insights, we formulate the optimal weight matrix design problem and show that, in the Gaussian case, it can be formulated as an SDP and hence efficiently solved. When the weight matrices are random, we prove that the rate functions of any node in the network lie between the rate functions corresponding to a fusion node, that processes all observations, and a node in isolation. The bounds hold for any random model of weight matrices, with the single condition that the weight matrices are independent from the agents' observations.

\begin{appendices}

\section{Proof of~Lemma~\ref{lemma-step-1}}
\label{app:finite-cover-compact-sets}
For every $\delta>0$, define $I^{\delta}: \mathbb R^d \mapsto \mathbb R$, $I^{\delta}(x):= \min\{I(x)-\delta, \frac{1}{\delta}\}$, and note that for any $D\subseteq \mathbb R^d$,
\begin{equation}
\label{delta-I-property}
\lim_{\delta\rightarrow 0} \inf_{x\in D}I^{\delta}(x)= \inf_{x\in D} I(x).
\end{equation}
Fix a compact set $F$. For every $y\in F$, choose $\lambda_y\in \mathbb R^d$ for which $\lambda_y^\top q - \Lambda(\lambda_y)\geq I^{\delta}(y)$~\footnote{Such a point must exist because of the following: If $I(y)$ is finite, then, since $I(y)$ equals the supremum $\sup_{\lambda \in \mathbb R^d} \lambda^\top q - \Lambda(\lambda)$, for every $\delta>0$, there must exist a point $\lambda^\prime=\lambda^\prime(\delta)$ such that ${\lambda^\prime}^\top y - \Lambda(\lambda^\prime)\geq I(y)-\delta$. Since $I(y)-\delta \geq I(y)$, taking $\lambda_y$ to be $\lambda^\prime(\delta)$ verifies the inequality. We can show in a similar way the existence of $\lambda_y$ in the case when $I(y)=+\infty$.}
Also, for each $y$ choose $\rho_y>0$ such that $\rho_y \|\lambda_y\|\leq \delta $.

Now, fix arbitrary $y\in F$. Then, by construction of $\rho_y$ and $\lambda_y$, we have:
\begin{equation*}
- \inf_{x\in B_{y}(\rho_y)} \lambda^\top x \leq  - \lambda_y^\top y + \delta.
\end{equation*}
Applying~\eqref{eq-it-suf-ball} for $D=B_{y}(\rho_y)$ and $\lambda=\lambda_y$ and combining it with the preceding equation yields
\begin{equation}
\label{eq-it-suf-ball-2}
\frac{1}{t}\log \mathbb P\left(X_{i,t}\in B_{y}(\rho_y)\right)\leq \delta - \lambda_y^\top y + \Lambda(\lambda_y).
\end{equation}
Extracting a finite cover $\{B_{y_i}(\rho_{y_i}): i=1,...,K\}$ of $F$ from the family of balls $\{B_{y}(\rho_y): y\in F\}$, and applying~\eqref{eq-it-suf-ball-2} to each of the balls,  we obtain by the union bound
\begin{equation*}
\frac{1}{t}\log \mathbb P\left(X_{i,t}\in F\right) \leq  \frac{1}{t}\log K + \delta - \min_{i=1,\ldots,K} \lambda_{y_i}^\top y_i -\Lambda(\lambda_{y_i}).
\end{equation*}
Recalling that for each $y$, $\lambda_y$ satisfies $\lambda_{y}^\top y -\Lambda(\lambda_{y})=I^\delta(y)$, and letting $t\rightarrow +\infty$,
\begin{equation*}
\limsup_{t\rightarrow +\infty}\frac{1}{t}\log \mathbb P\left(X_{i,t}\in F\right) \leq \delta - \min_{i=1,\ldots,K} I^{\delta}(y_i) \leq \delta - \inf_{y\in F} I^{\delta}(y).
\end{equation*}
Finally, letting $\delta\rightarrow 0$ and using the property~\eqref{delta-I-property} of $I^\delta$, the bound~\eqref{suffices-upper-bound} for compact sets follows.

\section{Proof of Exponential tightness of $\left\{\mu_{i,t}\right\}_{t=1,2,...}$}
\label{app:exponential-tightness}
This section proves Lemma~\ref{lemma-exp-tightness}.
Fix $i$ and, for each $t$ and $l$, $l=1,...,d$, define $\mu_{i,t}^l$ to be the probability measure on $\mathbb R$ induced by the $l$-th coordinate of vector $X_{i,t}$,
\[\mu_{i,t}^l\left( (-\infty,\rho]\right):=\mathbb P\left( X_{i,t}^l \leq \rho\right),\]
for $\rho\in \mathbb R$. For each $l$ let $\Lambda^l$ denote the log-moment generating function of $Z_{i,t}^l$, $\Lambda^l(\nu):=\log \mathbb E\left[e^{\nu Z_{i,t}^l} \right]$, $\nu \in \mathbb R$; note  that $\Lambda^l(\nu)= \Lambda (\nu e_l)$. Also, let $I^l$ denote the conjugate of $\Lambda^l$,
\begin{equation}
I^l(\rho)=\sup_{\nu\in \mathbb R} \rho \nu -\Lambda^l(\nu).
\end{equation}
Now, fix $\rho>0$. By the union bound, we have
\begin{equation}
\label{eq-complement-union-bound}
\mu_{i,t}\left( H_\rho^c\right) \leq \sum_{l=1}^d \mu_{i,t}^l((-\infty,-\rho])+ \sum_{l=1}^d \mu_{i,t}^l([\rho,+\infty]).
\end{equation}
We focus on the term on the right-hand side sum that corresponds to a fixed $l$. For any fixed $\nu \geq 0$, we have
\begin{equation*}
\mathbb P\left(X_{i,t}^l \geq \rho\right) \leq \mathbb E\left[e^{  t \nu X_{i,t}^l - t \rho \nu}\right].
\end{equation*}
Similarly as in eqs.~\eqref{eq-independence}, conditioning on $W_1,...W_t$, we obtain:
\begin{align*}
\mathbb E\left[\left.e^{  t \nu X_{i,t}^l}\right| W_1,\ldots,W_t\right]  &= \mathbb E\left[\left.e^{  \sum_{s=1}^t \sum_{j=1}^N \nu e_l^\top [\Phi(t,s)]_{ij}Z_{j,s} }\right| W_1,\ldots,W_t\right]\\
&=e^{\sum_{s=1}^t \sum_{j=1}^N \Lambda^l( [\Phi(t,s)]_{ij} \nu )}\\
&=e^{\sum_{s=1}^t \sum_{j=1}^N \Lambda( [\Phi(t,s)]_{ij} \nu e_l)},
\end{align*}
where the second equality follows by the fact that, given $W_1,\ldots,W_t$, terms $ \nu e_l^\top [\Phi(t,s)]_{ij}Z_{j,s}$ in the double sum above are independent. Applying now Lemma~\ref{simple-lemma} for $\lambda=\nu e_l$, and using the fact that $\Lambda^l(\nu)=\Lambda(\nu e_l)$ yields
\begin{equation*}
\mathbb E\left[\left.e^{  t \nu X_{i,t}^l}\right| W_1,\ldots,W_t\right] \leq e^{t \Lambda^l(\nu)}.
\end{equation*}
Combining the preceding three equations together with the monotonicity of the expectation, we obtain
\begin{equation}
\label{need-this}
\frac{1}{t}\log \mathbb P\left(X_{i,t}^l \geq \rho\right) \leq \Lambda^l(\nu) -  \rho \nu.
\end{equation}
We show that if $\rho > e_l^\top \theta= \theta_l$, the infimum of the right hand side of~\eqref{need-this} over all $\nu\geq 0$ equals $- I^l(\rho)$. To prove this, it suffices to show that if $\rho\geq \theta_l$, the supremum is not achieved for the negative values of $\nu$. Function $\Lambda^l$ is convex and differentiable for all $\nu$, and in particular at $\nu=0$ (as a log-moment generating function, see Lemma~\ref{lemma-properties-of-lmgf}). Thus, for any $\nu$, $\Lambda^l(\nu)\geq \Lambda^l(0)+ (\Lambda^l)^\prime(0)\nu= \theta_l \nu$. Thus, for $\nu <0$, we have $\rho \nu -\Lambda^l(\nu)\leq \nu (\rho-\theta_l)< 0$. Since we know that $I^l$ must be non-negative (see Lemma~\ref{lemma-properties-of-lmgf}), the claim above follows. Thus, for all $\rho\geq \theta_l$, we have:
\begin{equation}
\label{need-this-2}
\frac{1}{t}\log  \mu_{i,t}^l([\rho,+\infty]) \leq - I^l(\rho).
\end{equation}
By a similar procedure, one can also obtain that $\frac{1}{t}\log \mu_{i,t}^l((-\infty,-\rho])\leq - I_l(-\rho).$

Now, recall that by Assumption~\ref{ass-finite-at-all-points}, $\mathcal D_{\Lambda}=\mathbb R^d$; hence, $\mathcal D_{\Lambda^l}=\mathbb R$. Then, for any $\rho$
\begin{equation*}
I_l(\rho)=\sup_{\nu\in \mathbb R} \nu x - \Lambda_l(\nu) \geq  \nu |\rho| - \inf_{\nu: |\nu| \leq \nu_0} \Lambda_l(\nu),
\end{equation*}
where $\nu_0$ is an arbitrary positive number. Noting that the second term on the right hand side is finite, we see that $I_l$ grows unbounded as $|\rho|\rightarrow +\infty$. Since $l$ was arbitrary, we have that each of the exponents in~\eqref{eq-complement-union-bound} grows unbounded as $\rho$ increases to $+\infty$. This completes the proof of Lemma~\ref{lemma-exp-tightness}.

\section{Proof of Lemma~\ref{lemma-x-in-D^o-conditioning}}
\label{app:proof-of-lemma-x-in-D^o-conditioning}
Fix a measurable set $D$. We first show that if~\eqref{eq-it-suffices-to-consider-delta-balls} holds for any $x\in D^{\mathrm{o}}$ and any $\omega\in \Omega$, then for any $x\in D^{\mathrm{o}}$
\begin{equation}
\label{eq-sufficient-0}
\lim_{\delta\rightarrow 0}\liminf_{t\rightarrow +\infty}\frac{1}{t}\log \mathbb P\left( X_{i,t}\in B_x(\delta)\right) \geq - N I(x).
\end{equation}
To this end, fix $x\in D^{\mathrm{o}}$ and fix $\omega\in \Omega$.

Applying Fatou's lemma~\cite{Karr} to the sequence of random variables $R_t:=\frac{1}{t}\log \mathbb P\left(X_{i,t}\in D | W_1,\ldots,W_t\right)$, $t=1,2,\ldots$, we get
\begin{equation}
\label{eq-ineq-obtained-from-Fatous}
\liminf_{t\rightarrow +\infty}
\mathbb E \left[ \frac{1}{t}\log \mathbb P\left(X_{i,t}\in D | W_1,\ldots,W_t\right) \right]
\geq
\mathbb E \left[R^\star \right].
\end{equation}
where $R^\star(\omega):= \liminf_{t\rightarrow+\infty} R_t (\omega)$, $\omega\in \Omega$. Consider the left-hand side of~\eqref{eq-ineq-obtained-from-Fatous}. By linearity of the expectation and concavity of the logarithmic function, we have
\begin{align*}
 \mathbb E \left[ \frac{1}{t}\log \mathbb P\left(X_{i,t}\in D | W_1,\ldots,W_t\right) \right]
 \leq  \frac{1}{t}\log \mathbb E \left[ \mathbb P\left(X_{i,t}\in D | W_1,\ldots,W_t\right) \right]= \frac{1}{t}\log \mathbb P\left(X_{i,t}\in D\right).
\end{align*}
Taking the lim inf as $t\rightarrow +\infty$ on both sides of the preceding inequality and combining the result with~\eqref{eq-ineq-obtained-from-Fatous}, yields:
\begin{equation}
\label{eq-combining-concavity-with-Fatou}
\liminf_{t\rightarrow +\infty} \frac{1}{t}\log \mathbb P\left(X_{i,t}\in D\right)
\geq \mathbb E \left[R^\star \right].
\end{equation}
We now focus on the random variable $R_t$. Note that we assumed that $D^{\mathrm{o}}$ is nonempty (if the interior of $D$ is empty, the lower bound~\eqref{eq-universal-bounds-lower} holds trivially). Since $D^{\mathrm{o}}$ is open, for any $x \in D^{\mathrm{o}}$, we can find a small neighborhood $B_{x}(\delta_0)$ that is fully contained in $D^{\mathrm{o}}$ (where $\delta_0=\delta_0(x)$). Hence, for all $\delta \leq \delta_0$, we have $B_x(\delta) \subseteq D^{\mathrm{o}}\subseteq D$, and thus, for any fixed $\omega \in \Omega$
\begin{equation}
\label{eq-open-ball-inside-D-and-R_t}
R_t \geq \frac{1}{t}\log\mathbb P\left( X_{i,t}\in B_x(\delta)| W_1,...,W_t\right)
\end{equation}
(we used here that the logarithmic function is non-decreasing). Since~\eqref{eq-open-ball-inside-D-and-R_t} holds for all $t$ and all $\delta$ sufficiently small, taking the corresponding limits yields
\begin{equation*}
R^\star \geq \lim_{\delta\rightarrow 0}\liminf_{t\rightarrow +\infty} \frac{1}{t}\log\mathbb P\left( X_{i,t}\in B_x(\delta)| W_1,...,W_t\right).
\end{equation*}
Using now the assumption~\eqref{eq-it-suffices-to-consider-delta-balls} of the lemma to bound the right-hand side of the preceding inequality, we obtain $R^\star \geq - N I(x)$, which, we note, holds for every point $x$ in $D^{\mathrm{o}}$.
 Taking the supremum over all $x\in D^{\mathrm{o}}$, we obtain that for every $\omega \in \Omega$,
\begin{equation}
\label{eq-sup-over-D^o}
R^\star \geq - \inf_{x\in D^{\mathrm{o}}} N I(x).
\end{equation}
Taking the expectation in the left-hand side, and combining with~\eqref{eq-combining-concavity-with-Fatou}, we finally obtain the lower bound~\eqref{eq-universal-bounds-lower}:
\begin{equation*}
\label{eq-combining-concavity-with-Fatou}
\liminf_{t\rightarrow +\infty} \frac{1}{t}\log \mathbb P\left(X_{i,t}\in D\right)
\geq -\inf_{x\in D^{\mathrm{o}}} N I(x).
\end{equation*}
Since $D$ was arbitrary, the claim of Lemma~\ref{lemma-x-in-D^o-conditioning} is proven.

\section{Proof of Lemma~\ref{lemma-Lambda-t-M-is-nice}}
\label{app:proof-of-lemma-Lambda-t-M}
Being the sum of $\Lambda_t$ and a (convex) quadratic function, $\Lambda_{t,M}$ inherits convexity and differentiability from $\Lambda_t$ (in fact, $\Lambda_{t,M}$ is strictly convex due to strict convexity $\|\lambda\|^2/(2M)$). To prove $1$-coercivity, by convexity of $\Lambda_t$, we have that $\Lambda_t(\lambda)\geq \lambda^\top \theta$. Hence,
\begin{equation*}
\Lambda_{t,M}(\lambda)\geq \lambda^\top \theta + \frac{\|\lambda\|^2}{2M}.
\end{equation*}
Dividing both sides by $\|\lambda\|$ and using in the right hand side that $\lambda^\top \geq -\|\lambda\|\|\theta\|$, we obtain
\begin{equation*}
\frac{\Lambda_{t,M}(\lambda)}{\|\lambda\|} \geq -\|\theta\| + \frac{\|\lambda\|}{2M}\rightarrow +\infty,
\end{equation*}
when $\|\lambda\|\rightarrow +\infty,$ proving that $\Lambda_{t,M}$ is $1$-coercive. Strict convexity, differentiability, and $1$-coercivity of $\Lambda_{t,M}$ imply that the gradient map $\nabla \Lambda_{t,M}$ is a bijection, see, e.g., Corollary 4.1.3 in~\cite{Urruty}, p.~239. This proves the first part~\ref{lemma-Lambda-t-M-part-coercive}.

We now prove part~\ref{lemma-Lambda-t-M-part-unifbdd}. Fix $x$ and fix $t\geq 1$. Note that $\eta_t$ is the maximizer in $I_{t,M}(x)=\sup_{\lambda\in \mathbb R^d} \lambda^\top x - \Lambda_{t,M}(\lambda)$, and thus it holds that $I_{t,M}(x)=\eta_t^\top x - \Lambda_{t,M}(\eta_t)$. Since $\Lambda_t$ is convex (and differentiable), its gradient map is monotone. Hence,
\begin{equation}
\label{start-from-here}
\left(\nabla \Lambda_t(\eta_t) - \nabla \Lambda_t(0)\right)^\top (\eta_t-0)\geq 0.
\end{equation}
We next show that the value of the gradient of $\Lambda_{t}$ at $0$ equals $\theta$. From~\eqref{eq-Lambda-t-equals}, we have
\begin{equation}
\nabla \Lambda_{t}(\lambda) = \frac{1}{t}\sum_{s=1}^t\sum_{j=1}^N [\Phi(t,s)]_{ij} \Lambda([\Phi(t,s)]_{ij} \lambda).
\end{equation}
The gradient of $\Lambda$ at $\lambda=0$ equals $\theta$, see Lemma~\ref{lemma-properties-of-lmgf}. Using the fact that, for each fixed $s$, $\sum_{j=1}^N[\Phi(t,s)]_{ij}=1$, we obtain that $\nabla \Lambda_t(0)=\theta$. Thus, from \eqref{start-from-here} we have
\begin{equation}
\label{continue}
\left(\nabla \Lambda_t(\eta_t) - \theta\right)^\top \eta_t\geq 0.
\end{equation}
Now, note from~\eqref{eq-Lambda-t-equals} that $\nabla \Lambda_t(\lambda) = \nabla \Lambda_{t,M}(\lambda)-\lambda/M$, for arbitrary $\lambda$. Using now the fact $\nabla \Lambda_{t,M}(\lambda)=x$,~\eqref{continue} implies $(x-1/M\eta_t - \theta)^\top \eta_t\geq 0.$
Thus, $(x-\theta)^\top \eta_t \geq \eta_t^\top \eta_t/2$, proving the claim of the lemma for this fixed $t$ and $x$. Since these were arbitrary, the proof of the lemma is complete.

\section{Proof of Lemma~\ref{lemma-widetilde-I-t-M}}
\label{app:proof-xi}
From the fact that $\mathcal D_{\widetilde \Lambda_{t,M}}=\mathbb R^d$, one can show that $\widetilde I_{t,M}$ has compact level sets (note that $\widetilde I_{t,M}$ is lower semicontinuous).
Thus, the infimum in~\eqref{eq-if-we-manage-we-prove} has a solution.
Denote this solution by $w_t$ and let $\zeta_t$ denote a point for which
$w_t= \nabla \widetilde \Lambda_{t,M}\left(\zeta_t\right) \left(= \nabla\Lambda_{t,M}\left(\zeta_t+\eta_t\right)\right)$
(such a point exists by Lemma~\ref{lemma-Lambda-t-M-is-nice}).
We now show that $\left\|w_t \right\|$ is uniformly bounded for all $t$,
 which, combined with part~\ref{lemma-Lambda-t-M-part-unifbdd} of Lemma~\ref{lemma-Lambda-t-M-is-nice}, in turn implies that $\eta_t+\zeta_t$ is uniformly bounded.
 \begin{lemma}
 \label{lemma-w-t-zeta-t-bounded}
 For any fixed $\delta>0$ and $M>0$, there exists $R=R(x,\delta,M)<+\infty$ such that for all $t$:
 \begin{enumerate}
\item $\left\|w_t\right\|\leq R$, and
\item $\left\|\zeta_t + \eta_t\right\| \leq M (R+ \|\theta\|)$.
\end{enumerate}
 \end{lemma}
\begin{proof}
Fix $M>0, \delta>0$. Define $\overline f_M,\underline f_M: \mathbb R^d\mapsto \mathbb R$ as
\begin{align*}
\overline f_M(z) &=\sup_{\lambda \in \mathbb R^d} \lambda^\top z - N \Lambda\left( 1/N \lambda\right)- \frac{\|\lambda\|^2}{2M},\\
\underline f_M(z)&=\sup_{\lambda \in \mathbb R^d} \lambda^\top z - \Lambda(\lambda)- \frac{\|\lambda\|^2}{2M},
\end{align*}
for $z\in \mathbb R^d$. Note that both $\overline f_M, \underline f_M$ are lower semicontinuous, finite for every $z$, and have compact level sets.
Let $c=\inf_{z\in B^{\mathrm c}_x(\delta)}\overline f_M (z)<+\infty$,
and define $S_{c}=\left\{z\in \mathbb R^d: \underline f_M(z)\leq c\right\}$.

Fix arbitrary $t\geq 1$. One can show, with the help of Lemma~\ref{simple-lemma}, that, for any $z\in \mathbb R^d$,
\begin{equation}
\label{eq-overline-f-M-underline-f-M}
\underline f_M(z) \leq I_{t,M}(z) \leq \overline f_M(z).
\end{equation}
Observe now that
$I_{t,M}(w_t)= \inf_{z\in B^{\mathrm c}_x(\delta)} I_{t,M}(z)\leq \inf_{z\in B^{\mathrm c}_x(\delta)} f_M(z)\leq c$.
On the other hand, taking in~\eqref{eq-overline-f-M-underline-f-M} $z=w_t$, yields
 $\underline f_M(w_t) \leq I_{t,M}(w_t)$, and it thus follows that $w_t$ belongs to $S_c$.

Finally, as $S_c$ is compact, we can find a ball of some radius $R=R(x,M,\delta)>0$ that covers $S_c$, implying $w_t\in B_0(R)$.
Since $t$ was arbitrary, the claim in part~1 follows.

We now prove part~2. Recall that, for any $t$, $w_t$ and $\zeta_t+\eta_t$ satisfy $w_t=\nabla \Lambda_{t_0,M}\left( \zeta_t+\eta_t\right)$.
Applying part~\ref{lemma-Lambda-t-M-part-unifbdd} of Lemma~\ref{lemma-Lambda-t-M-is-nice} for $z=w_t$, we have that $\left\| \zeta_t+\eta_t\right\|\leq M \left\| w_t - \theta\right\|$.
Combining this with part~1 of this lemma,
\begin{equation*}
\left\| \zeta_t+\eta_t\right\|\leq M \left\| w_t - \theta\right\|
\leq M \sup_{w\in B_0(R)}\left\| w - \theta\right\|\leq M(R+\|\theta\|).
\end{equation*}
This completes the proof of part~2 and the proof of Lemma~\ref{lemma-w-t-zeta-t-bounded}.
\end{proof}

Fix $x,\delta$ and $M$ and define $r_1= M \left\|z-\theta \right\|$, $r_2= M (R+\|\theta\|)$,
where $R$ is the constant that verifies Lemma~\ref{lemma-w-t-zeta-t-bounded}.
Fix now $t\geq 1$ and recall that $\eta_t$, $\zeta_t$, and $w_t$ are chosen such that
$x=\nabla \Lambda_{k,M}\left( \eta_t\right)$,  $\widetilde I_{t,M} \left(w_t\right)=\inf_{z \in B^{\mathrm{c}}_x(\delta)} I_{t,M}(z)$,
and $w_t=\nabla \Lambda_{t,M}\left( \eta_t+ \zeta_t\right)$.
By part~\ref{lemma-Lambda-t-M-part-unifbdd} of Lemma~\ref{lemma-Lambda-t-M-is-nice} and part~2 of Lemma~\ref{lemma-w-t-zeta-t-bounded} we have for $\eta_t$ and $\zeta_t$,
$\left\|\eta_t\right\|\leq r_1$,  $\left\|\eta_t+\zeta_t\right\|\leq r_2$. To prove Lemma~\ref{lemma-widetilde-I-t-M}, we first show that there exists some
positive constant $r_3$, independent of $t$, such that $\left\|\zeta_t\right\|\geq r_3$ for all $t$. To this end, consider the gradient
map $\lambda \mapsto \nabla \Lambda_{t,M}(\lambda)$, and note that $\nabla \Lambda_{t,M}$ is continuous, and hence
uniformly continuous on every compact set. Note also that
$\left\|\eta_t\right\|,\left\|\eta_t+\zeta_t\right\|\leq \max\{r_1,r_2\}$; that is, points $\eta_t$ and $\eta_t+\zeta_t$ are uniformly bounded for all $t$.
Suppose now, for the sake of contradiction, that for some sequence of times $t_k$, $k=1,2,...$, $\left\|\zeta_{t_k} \right\|\rightarrow 0$, as $k\rightarrow +\infty$. Then,
 $\left\|\left(\eta_{t_k}+ \zeta_{t_k}\right) - \eta_{t_k}\right\|\rightarrow 0$, and hence, by the uniform continuity of $\nabla \Lambda_{t,M}(\cdot)$ on $\overline B_{0}(\max\{r_1,r_2\})$ we have
 \[\left\|\nabla \Lambda_{t,M}(\eta_{t_k}) - \nabla \Lambda_{t,M}(\eta_{t_k}+\zeta_{t_k})\right\|\rightarrow 0,\:\:\mathrm{as}\,\,t\rightarrow \infty.\]
 %
Recalling that
$x=\nabla \Lambda_{t,M}\left(\eta_{t_k}\right)$, $w_{t_k}= \nabla \Lambda_{t,M}\left(\eta_{t_k}\right)$, yields
\[\left\|w_{t_k}-x\right\|\rightarrow 0.\]

This contradicts with the fact that, for all $t$, $w_t \in B^{\mathrm{c}}_{x} (\delta)$. Thus, we proved the existence of $r_3$ independent of $t$
such that $\left\|\zeta_t\right\| \geq r_3$, for all $t$.

Now, let
\[\Upsilon=\left\{(\eta,\zeta) \in \mathbb R^d\times \mathbb R^d:
\|\eta\|\leq r_1,  \|\eta + \zeta \| \leq r_2, \|\zeta\|\geq r_3\right\},\]
and introduce $g:\mathbb R^d\times\mathbb R^d\mapsto \mathbb R$,
\begin{equation}
g(\zeta,\eta)= \Lambda_{t,M}(\eta) - \Lambda_{t,M}(\zeta+ \eta) + \nabla \Lambda_{t,M}(\zeta+ \eta)^\top \zeta.
\end{equation} %
By strict convexity of $\Lambda_{t,M}(\cdot)$, we see that, for any $\eta$ and $\zeta \neq 0$,
the value $g(\eta,\zeta)$ is strictly positive.
Further, note that since $\Lambda_{t,M}$ and $\nabla \Lambda_{t,M}$ are continuous, function $g$ is also continuous.
Consider now
\begin{equation}
\label{eq-infimum-of-g}
\xi:=\inf_{(\eta,\zeta)\in\Upsilon} g(\eta,\zeta).
\end{equation}
Because $\Upsilon$ is compact, by the Weierstrass theorem, the problem in~\eqref{eq-infimum-of-g} has a solution, that is,
there exists $(\eta_0,\zeta_0)\in \Upsilon$, such that $g(\eta_0,\zeta_0)=\xi$.
Finally, because $g$ is strictly positive at each point in $\Upsilon$ (note that $\zeta\neq 0$ in $\Upsilon$),
we conclude that $\xi= g(\eta_0,\zeta_0)>0$.

Returning to the claim of Lemma~\ref{lemma-widetilde-I-t-M}, by Lemma~\ref{lemma-w-t-zeta-t-bounded},  $\left(\eta_t,\eta_t+\zeta_t\right)$ belongs to $\Upsilon$, and, thus,
\begin{align*}
\widetilde I_{t,M} (w_t)& = \Lambda_{t,M}(\eta_t) - \Lambda_{t,M}(\zeta_t+ \eta_t)  +
\nabla \Lambda_{t,M} (\zeta_t+ \eta_t)^\top \zeta_t \\
&= g\left(\eta_t, \zeta_t\right) \geq \xi.
\end{align*}
This completes the proof of Lemma~\ref{lemma-widetilde-I-t-M}.

\end{appendices}

\bibliographystyle{IEEEtran}
\bibliography{IEEEabrv,Bibliography}

\end{document}